\documentclass[10pt,journal]{IEEEtran}
\IEEEoverridecommandlockouts
\usepackage[letterpaper, left=0.667in, right=0.667in, bottom=0.597in, top=0.792in]{geometry}
\usepackage{amsfonts,mathptmx}
\usepackage{textcomp}
\usepackage{savesym}

\usepackage{changes}
\savesymbol{comment}

\usepackage{url,cite}
\usepackage{multirow}
\usepackage{mathtools}
\usepackage{hyperref}
\usepackage{comment}
\usepackage{xcolor}
\usepackage{amsmath,graphicx,latexsym,amssymb,amsthm}
\newcommand{\mR}{\mathbb{R}}

\newcommand{\mc}{\mathcal}

\newcommand{\mE}{\mathbb{E}}

\newcommand{\tha}{\theta}

\renewcommand{\leq}{\leqslant}
\renewcommand{\geq}{\geqslant}

\newtheorem{theo}{Theorem}

\newtheorem{prop}{Proposition}
\newtheorem{Def}{Definition}
\newtheorem{assume}{Assumption}

\newtheorem{lemma}{Lemma}

\usepackage{algpseudocode}
\usepackage{algorithmicx}
\usepackage{algorithm}
\usepackage{caption}
\usepackage{eufrak}

\DeclareMathOperator{\supp}{supp}

\begin{document}

\title{Game-Theoretic Neyman-Pearson Detection to Combat Strategic Evasion}
\author{Yinan Hu, {\it Student Member}, Juntao Chen, {\it Member}, and Quanyan Zhu, {\it Senior Member}
\thanks{Yinan Hu and Quanyan Zhu are with the Department of Electrical and Computer Engineering, New York University, Brooklyn, New York, USA, 11201. Email: \{yh1817,qz494\}@nyu.edu}
\thanks{Juntao Chen is with the Department of Computer and Information Sciences, Fordham University
New York, USA 10023. Email: jchen504@fordham.edu }}
\date{\today}
\maketitle
\begin{abstract}
\noindent The security in networked systems depends greatly on recognizing and identifying adversarial behaviors. Traditional detection methods target specific categories of attacks and have become inadequate for increasingly stealthy and deceptive attacks that are designed to bypass detection strategically. This work proposes game-theoretical frameworks to recognize and combat such evasive attacks. We focus on extending a fundamental class of statistical-based detection methods based on Neyman-Pearson's (NP) hypothesis testing formulation. We capture the conflicting relationship between a strategic evasive attacker and an evasion-aware NP detector. By analyzing both the equilibrium behaviors of the attacker and the NP detector, we characterize their performance using Equilibrium Receiver-Operational-Characteristic (EROC) curves. We show that the evasion-aware NP detectors outperform the non-strategic ones in the way that the former can take advantage of the attacker’s messages to adaptively modify their decision rules to enhance their success rate in detecting  {abnormal} attacks. In addition, we extend our framework to a sequential setting where the user sends out identically distributed messages. We corroborate the analytical results with a case study of intrusion detection evasion problem.  \\
\end{abstract}
\setcounter{page}{1}

\section{Introduction}
Cyber-security is the practice of protecting critical systems and their sensitive information from cyber attacks. There has been a surging proportion of evasion attacks among all cyber attacks over the years. A recent WatchGuard lab report \cite{2021evasion_report} reveals that over 78\% of the malware attacks with encryption are evasive in 2021.  
Defending against evasion attacks has been a crucial concern for networked systems such as computer networks \cite{jain2009security_computer_networks}, Internet of Things (IoT) \cite{tuptuk2018manufacturing_systems}, and autonomous systems \cite{kim2021cybersecurity_autonomous}. 

Many studies \cite{chandola2009anomaly_detection}\cite{kruegel2005automating_mimicry} focused on detecting and mitigating network attacks. 
However, as attackers become more intelligent, deceptive, and stealthy, they are increasingly able to evade detection and move laterally within networks to target critical assets.
These existing countermeasures are often inadequate because they fail to account for evasion as a strategic behavior by attackers. It is highly likely that attackers can learn from updated detection mechanisms and adapt their strategies to further infiltrate networked systems. Therefore, there is a pressing need to fundamentally improve detection methods by developing evasion-aware techniques that proactively strengthen network security. Such methods need to assume that attackers are rational, omniscient, and always acting strategically to evade detection.

{In this work, we aim to build on existing research \cite{hu2022evasion} by developing an evasion-aware detection theory for stealthy attacks. Our focus is on a fundamental class of Neyman-Pearson (NP) hypothesis-testing problems \cite{neyman_pearson1933}. In these scenarios, a user (referred to as ``It") communicates with a detector (referred to as ``She") by sending messages. The detector is aware that the user can be either normal or abnormal, and it is assumed that different types of users generate messages following distinct distributions, reflecting their respective incentives. The detector’s primary goal is to infer the user’s type based on the received message samples. We employ a game-theoretic approach to model the strategic interactions between the detector and an evasive attacker (referred to as ``He"), who possesses complete knowledge of both normal and abnormal behaviors as well as the detection mechanism.}  We model the detection-evasive attacker as an adversary who stands in between the user and the detector, knows the user's true type, intercepts the user's original messages as inputs, and strategically designs outputs sent to the detector to undermine the latter's performance. We formulate sequential game problems to capture the capability and the incentive of the attacker and the interactions with nature and the detector. The nature  first selects a user from the population with a prior distribution on the type. The user then generates messages based on its type. Based on the user's type and the messages, the attacker generates distorted messages and sends them to the detector, who finally determines the user's type based on those distorted messages. 

The analysis of the equilibrium of the game frameworks provides essential insights into fundamental limits on the performance of evasion-aware detectors and their designs. We observe that by varying the intrinsic coefficients parameterizing the detection strategies, one can compute relationships between detection rates and false alarm rates and depict them in terms of an Equilibrium Receiver Operational Characteristics (EROC) curve \cite{kruegel2003bayesian_IDS,lee1998data_IDS,lee1999data_IDS}), which characterizes the fundamental performance tradeoffs of the detector. Furthermore, The analysis of game-theoretic equilibrium also provides numerical algorithms to compute the best-response strategies for both the attacker and the detector.  

We study game models for two kinds of detectors depending on their reasoning capabilities. One is passive, and the other is proactive. The passive one does not anticipate the presence of the attacker and treats the inputs as if they were original, untainted ones from the user. Her decision process based on the inputs is naive, non-strategic and thus her detection performance is subject to the adversarial poisoning of the inputs. In contrast, the proactive detector anticipates that an attacker can manipulate the inputs and checks whether the inputs are consistent with the history and the known priors of the hypotheses. Once she suspects that an attacker exists, she adopts a detection strategy that takes into account the attacker's manipulation. She also takes advantage of the distorted inputs to infer the attacker's private information to mitigate the detrimental effects of poisonous inputs on her detection performance. On the other hand, the attacker decides his distortion strategy via balancing the trade-off between undermining the detector's performance and lowering the penalty resulting from manipulation of messages.   

The passive NP detector game is modeled by the Stackelberg game \cite{bacsar1998dynamic_game}. The attacker is the leader in the game who determines a distortion strategy while {the passive NP detector is the follower who naively applies classical NP hypothesis testing method for detection without knowing the distortion.} The analysis of the Stackelberg equilibrium shows that the detection rate of a passive NP detector is reduced exponentially with respect to the distortion level, compared to the attacker-free detection scenario. The attacker's distortion policy depends on the user's type: when the user is  {normal}, he does not mutate the original message; when the user is  {abnormal}, he manipulates the message by suppressing the probability that outgoing message lies in the region where the detector can make a detection. By depicting the EROC curve, we observe that in general the passive NP detector's performance is worse than the one of non-adversarial detector yet also dependent on the penalization: when the attacker's penalty goes to infinity, the attacker conveys the message truthfully, and the detector applies classical NP hypothesis testing to determine the user's type; in contrast, when the penalty vanishes, the detection rate reduces to zero as the attacker can distort the user's messages in any desired way.      

The proactive NP detector game is modeled using a signaling game framework \cite{cho1987signaling_game_stable_equil}. The attacker is the sender who can observe the underlying true hypothesis and its associated distribution of messages. The detector is the receiver who does not know the true hypothesis and makes a decision based on the message sent by the attacker who aims to mislead the decision. The detector is aware of the possible presence of the attacker and proactively infers true hypothesis using a Bayesian approach. The analysis of the equilibrium shows that the attacker's policies change not only the probability distributions of messages at different regions but the detector's regions of rejection as well. The attacker enlarges the proactive detector's uncertainty area significantly by interpolating the two original distributions followed by a scaling factor to make them less distinguishable. 
{The EROC curves show that the proactive NP detector outperforms the passive NP detector by leveraging the attacker's strategies.}

The results from the game-theoretic models have strong implications in building the next-generation evasion-aware detection systems. In particular, when the user is  {abnormal}, the attacker's strategy at equilibrium against a proactive detector has a stronger amplitude in the detector's region of rejection than the counterpart against a passive detector. So an evasion-aware detection system outperforms a classical counterpart in the way that the former can predict the regions of messages leading to a decision of  {abnormal} user and alter its decision rule so as to mitigate the attacker's distortion on the region of rejection, improving the chance of catching attacks that are supposed to be evasive. We investigate a case study of an anomaly detection system where an attacker distorts a user's event history to compromise the detection result. We use this study to corroborate the results and demonstrate the capability of evasion-aware detection.

Our contributions are two-fold. First, we develop a holistic theory for evasion-aware detection. In particular, the proposed framework is solidly built on classical hypothesis testing and game-theoretic theories and characterizes the fundamental performance tradeoffs for passive and proactive detectors. Second, evasion-aware detection policies obtained from the analysis of sequential games provide a design paradigm for next-generation detection systems. Such techniques provide a novel way for the detector to process the poisonous inputs while capable of maintaining a high detection performance.  

The rest of this work is organized as follows. 
In Section \ref{sec:stackelberg_game}, 
we present the Stackelberg game formulation of a passive detector, who is na\"ive about the distortion of messages by an attacker-in-the-middle.  
In Section \ref{sec:signalling}, we discuss the signaling game formulation of a proactive detector who anticipates adversarial behaviors. We model the relationship between the attacker-in-the-middle and the proactive detector using signaling games. The attacker applies a decision rule based on her posterior belief on the user's type after receiving distorted messages. In section \ref{sec:repeat_obs} we discuss performances of detectors under repeated observations.  In Section \ref{sec: numerical_experiment} we apply the results to a case study of an  IDS evasion problem. Finally, we conclude in Section \ref{sec:conclusion}.

\subsection{Related Works}

Many studies have contributed to developing stealthy attacks and their countermeasures in networked systems \cite{etesami2019dynamic_game_cps_review,mao2020stealthy_attack_ncs, pasqualetti2013attack_cps_detection}. One example is false data injections \cite{ashfaq2013ADS_evasion_margin}, where an attacker can inject stealthy data into the original traffic in the networked system. Another example is deception attacks \cite{Kang2017stealthy_deception_CPS}, where the attacker sends falsified data that the detector consider genuine and untainted. The concept of developing our proposed frameworks can be adapted to develop detectors combatting strategically stealthy attacks of any categories. 

Game-theoretic formulations have been widely used to capture the attacker-defender interactions in cyber-security \cite{van2013_stealthy_take_over,  manshaei2013game_network_security,zhu2015game_CPS,kiennert2018game_IDS,basar2016dynamic_cps}. A recent survey investigates how game-theoretic approaches are applied in intrusion detection systems. Authors in \cite{zhu2015game_CPS} introduce game-theoretic frameworks in formulating the tradeoffs among robustness, security, and resilience for cyber-physical control systems. {The authors in \cite{xiao2018prospect_theory_detection_APT} consolidate cumulative prospect theory into the game theory to model the complex interactions between cyber systems and advanced persistent threat (APT) attackers. They uncover the dynamics of risk-seeking and risk-averse behaviors by examining how subjective perceptions influence choices regarding attack and scan intervals.} Works in \cite{chen2009game_IDS_heter} study the IDS with multiple attackers and heterogeneous defenders in a network by formulating their relationships into a game, where their different payoff functions capture the heterogeneity among defenders. This work studies the attacker-defender relationship using a formulation combining the formulations of game theories and hypothesis testing similar to \cite{saritacs2019hypothesis_game}, yet we adopt more involved utility functions. 

Machine learning techniques have also been widely applied in detection and prediction problems \cite{basar2010network_security}. The main difference between detection by binary hypothesis testing and detection by machine learning method is that the former is an inference method, while the latter is a prediction method.
The security of machine learning has also been a research spotlight  \cite{dalvi2004adversarial_classification,huang2011adv_ML,barreno2006adv_ML_start,biggio2013evasion_attack_ML} for the past several decades. Evasion attacks on machine learning mainly include modifying the training or testing data, obfuscating the gradient descent steps in the optimization process, etc.  Authors in \cite{baras2006classifier_security} discuss the optimal evasion attacks upon the gradients at testing time in classification problems. 

  {Authors in \cite{bruckner2012static_adversarial_game} formulate an adversarial learning scenario where a data generator selects data samples and a learner selects the optimal parameters for the learning model into a static Nash prediction game, in which the uniqueness of the equilibrium is proved. They later propose a Stackelberg game formulation to characterize the scenario where testing data are strategically generated in response to the predictive model\cite{bruckner2011stackelberg_prediction}, which is applicable in email spam filtering.  
Authors in \cite{barreno2010security_machine_learning} discuss the vulnerability in machine learning techniques and categorize several classes of attacks upon machine learning systems along three axes: in terms of the influence, there are causative attacks or exploratory attacks; in terms of security violation, there are integrity attacks or availability attacks; in terms of specificity, there could be targeted attacks or indiscriminate attacks. Defensive measures are discussed under each category of attacks. The proposed framework extends hypothesis testing formulations and thus is model-based, while the machine-based detection techniques are data-driven.}
{The key advantage of a model-based approach is its reliance on a well-established theoretical framework, offering better interpretability and robustness. Unlike data-driven methods, which require large datasets and may struggle with generalization, model-based methods are more effective when data is limited or when system behavior can be mathematically modeled.}



\section{Problem Formulation of Passive Detectors}
\label{sec:stackelberg_game}
In this section, we present the problem formulation of a passive detector. As a comparison, we first consider the non-adversarial scenario that involves a user (It), a detector (She), and a communicating channel without an attacker in the middle (He). There are two types of the user:  {normal} and  {abnormal}. A user is sampled from the population with a known prior distribution of the two types of users. The sampled user knows its own type. Different types of users generate messages obeying distinct distributions. The detector does not know the user's type, but she can infer the type using the user's messages using hypothesis testing.

Let $\langle M,\mc{M}\rangle$ be a measurable Polish space endowed with $F_0,F_1\in \Delta(M)$ as two measures of
interest, where $\Delta(M)$ refers to the set of all measures supported on $M$. Denote $f_0,f_1\in L^1(M)$ as the probability density functions (strictly speaking, they are Radon-Nikodym derivatives of $F_0,F_1$ with respect to the Lebesgue measure) corresponding to the measures $F_0,F_1$. Suppose a {normal} user and an {abnormal} user generate messages obeying the distributions of $f_0,f_1$, respectively. Denote $\tilde{m}'$ as the random variable of the user's messages whose realizations are supported on $M$. Denote $\mc{H} = \{H_0,H_1\}$ as the space of user's types. Let $H_0$ be the hypothesis that the user is  {normal} and $H_1$ be the one that the user is  {abnormal}. Under different hypotheses, the messages are generated using $f_0$ and $f_1$, respectively:  
\begin{equation}
    H_0:\;\tilde{m}'\sim f_0(m'),\;\longleftrightarrow H_1:\tilde{m}'\sim f_1(m'). 
    \label{def:hypos}
\end{equation}

Let $\Gamma$ be the set of all decision rules of the detector, and a typical decision rule is denoted as $\delta\in\Gamma,\;\delta:M\rightarrow \{0,1\}$, where $\delta(m')=i,\;i\in\{0,1\},$ means that the detector determines the user is of type $H_i$. We suppose that the detector makes decisions using NP hypothesis testing method \cite{lehmann2006testing_hypothesis}; i.e., the detector arrives at an optimal policy $\delta^*$ by solving the following optimization problem:
\begin{equation}
\begin{aligned}
    \underset{{\delta} \in \Gamma}{\max}  \int_{\delta(m')=1}{f_1(m')dm'},\;\text{s.t.}\;\int_{\delta(m')=1}{f_0(m')dm'} \leq \alpha.
\end{aligned}
\end{equation}

The formulation of Neyman-Pearson hypothesis testing stresses the asymmetry of risks between false positive error and false negative error. The meaning of these types of errors may vary depending on the meaning of hypotheses $H_0$, $H_1$, in which here we assume that $H_1$ is for abnormal users. Thus, the formulation means that the detector aims to minimize the misdetection rate while constraining the false alarm rate.  However, in some other scenarios, such as vulnerability detection in web services \cite{Antunes2015web_service_vulnerability}, failing to catch an abnormality is more serious than a false alarm. Thus, the detector should aim to constrain the misdetection while minimizing the false alarm rate. In that case, we may as well denote $H_0$ as abnormal traffic and $H_1$ as normal traffic.  

\subsection{The Threat Model}
{
 In an adversarial scenario, the attacker aims to undermine the detector's detection rate as much as possible.  The attacker knows whether the user is normal or abnormal and has full knowledge of the detector's operations. Specifically, the attacker is aware that the detector employs the Neyman-Pearson hypothesis testing method to classify the user's type based on the received message.
The attacker intercepts the original messages from the user, manipulates them, and sends the distorted messages to the detector.  Upon receiving the message, the detector determines whether it originates from a normal or abnormal user. The passive detector, however, is non-strategic and unaware of the attacker's presence. }

Mathematically, we denote the user's original message as $m'$. Based on the user's type $H_i,\;i\in\{0,1\}$, the attacker manipulates the message through a mapping $\mu_i: M\rightarrow M,\;m'\longmapsto m$ with $m = \mu_i(m')$, where $m$ refers to the distorted message received by the detector (see the blue arrows in Figure \ref{fig:naive_detector}). Finally, the detector applies NP hypothesis testing \cite{neyman_pearson1933} to $m$ with distributions $f_1,f_0$ to determine the true type of the user. We capture the interaction between the attacker and a passive detector using a Stackelberg game \cite{bacsar1998dynamic_game}. 
\begin{figure}
    \centering
    \includegraphics[scale=0.28]{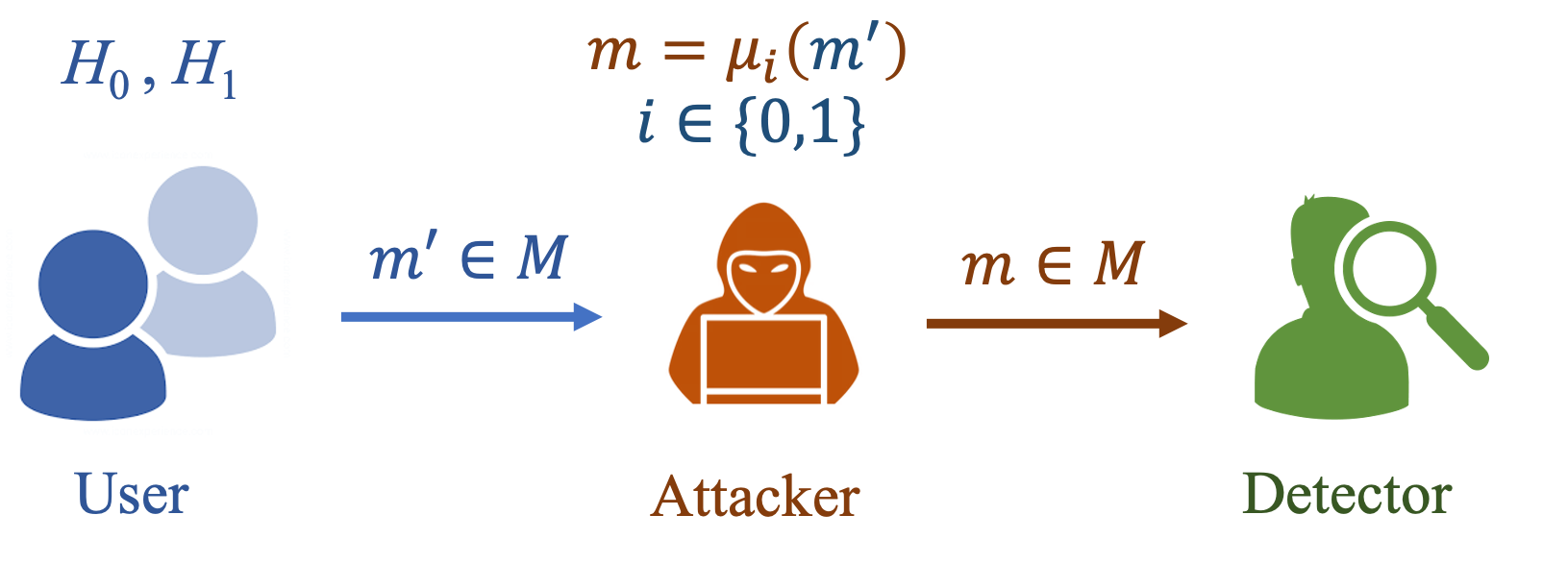}
    \caption{The adversarial detection with an attacker in the middle.  The nature first draws a message $m'\in M$ based on a randomly chosen hypothesis from either the null $H_0$ or the alternative $H_1$ hypothesis. It can be interpreted as the scenario where a user of known type (e.g.,  {normal} or  {abnormal}) is drawn from a prior distribution. An attacker observes the type or the hypothesis $i\in\{0,1\}$ and the associated message $m'$. The attacker chooses a strategy $\mu_i$ to distort the message and send $m=\mu_i(m')$ to the detector with the goal to mislead the detection result. The detector relies on her knowledge and designs a passive or proactive detection strategy to optimize her detection performance. 
    }
    \label{fig:naive_detector}
\end{figure}

\subsection{Detector's Decision Rule}
Let $\hat{\delta}\in\Gamma,\;\hat{\delta}: M \rightarrow \{0,1\}$ be a generic strategy of a passive detector, where $\hat{\delta}(m) = j,\;j\in\{0,1\}$ means that she thinks the user is of type $H_j\in\mc{H}$ based on the message $m$. 
Similar to the non-adversarial scenario, the detector arrives at her optimal strategy $\hat{\delta}^*$ by solving the following optimization problem:
\begin{equation}
\begin{aligned}
    &\underset{\hat{\delta} \in \Gamma}{\max}  \int_{\hat{\delta}(m)=1}{f_1(m)dm},\\
    &\text{s.t.}\;\int_{\hat{\delta}(m)=1}{f_0(m)dm} \leq \alpha, 
\end{aligned}
\end{equation}
where $\int_{\hat{\delta}(m)=1}{f_1(m)dm}$ and $\int_{\hat{\delta}(m)=1}{f_0(m)dm}$ refer to the `conjectural detection rate' and the `conjectural false alarm rate' from the receiver's viewpoint. 
 During the gameplay, the attacker-in-the-middle distorts the message via $\mu_0$ or $\mu_1$, depending on the observed user's type. Notice that we can regard the distortion mappings equivalently as the transformation from the original measures of interest $F_0,F_1$ to the final measures $\Sigma_0,\Sigma_1$. We argue via the following lemma from  \cite{villani2009ot_thick} that it is also equivalent to formulate the distortion using the resulting measures $\Sigma_0,\Sigma_1$:  
\begin{lemma}
Let $\langle M,\mc{M} \rangle$ be a Polish probability space endowed with atomless measures $F_0,F_1\in\Delta(M)$. Then, for any atomless probability measures $\hat{\Sigma}_0,\hat{\Sigma}_1\in \Delta(M)$, there exist measurable isomorphisms $\mu_0,\mu_1: M\rightarrow M$ such that 
\begin{equation}
\begin{aligned}
  \mu_0 \# F_0 = \Sigma_0,\;\;\mu_1 \# F_1 = \Sigma_1, \;  \mu^{-1}_0\#\Sigma_0 = F_0,\;\;\mu^{-1}_1 \# \Sigma_1 = F_1.
    \label{pushforward}
\end{aligned}
\end{equation}
\label{lemma:mu_sigma_equiv}
In other words, the measures $\Sigma_0,\Sigma_1$ are push-forward measures \cite{folland1999real_analysis} of the measures $F_0,F_1$, respectively.
\end{lemma}

Lemma \ref{lemma:mu_sigma_equiv} states that characterizing the attacker's strategies by $\hat{\mu}_1,\hat{\mu}_0$ is equivalent to characterizing them by the measures of distorted messages $\hat{\Sigma}_1,\hat{\Sigma}_0$, the latter of which is equivalent to the probability distributions $\hat{\sigma}_1,\hat{\sigma}_0$ on $M$. Thus, we define the attacker's action space $\hat{A}_S$ as follows:
\begin{equation}
\begin{aligned}
   \hat{A}_S := &\Big\{\begin{bmatrix}
    \hat{\sigma}_0 \\ \hat{\sigma}_1
    \end{bmatrix}\in {L}^1(M)^{\otimes 2}:
    i\in\{0,1\},\; \int_{m\in M}{\hat{\sigma}_i(m)dm} = 1; \\
    &\;\forall m\in M,\hat{\sigma}_i(m)\geq 0\Big\}.
    \label{eq:strategy_space_sender}
\end{aligned}
\end{equation}

The attacker has two objectives. On one hand, he seeks to undermine the detector's performance; on the other hand, he aims at lowering the `lying cost' from the distortion of messages. To evaluate the detector's performance, we introduce the following detection rate $\hat{P}_D: \Gamma \rightarrow [0,1]$ and the false alarm rate $\hat{P}_F:\;\Gamma \rightarrow [0,1]$:
\begin{align}
    \hat{P}_D(\hat{\delta}) &= \int_{\hat{\delta}(m)=1}{\hat{\sigma}_1(m)dm},\;\;  \label{detection_rate_stackelberg}\\ \hat{P}_F(\hat{\delta}) &= \int_{\hat{\delta}(m)=1}{\hat{\sigma}_0(m)dm}.
     \label{false_alarm_rate_stackelberg}
\end{align}
The attacker also incurs a `lying cost' \cite{kartik2009lying_cost} for the discrepancy between the distorted and the original distribution. Such a lying cost results from anomaly detection techniques in security \cite{ali2015intrusion_detection_evasion}.  
Considering the KL divergence as the penalty function,  we formulate the attacker's optimization problems as follows. For $j\in\{0,1\}$ and $(\sigma_0,\sigma_1)\in\bar{A}_S$, 
\begin{equation}
\begin{aligned}
\underset{{\sigma_j}}{\min}&\;{\bar{C}_S(H_j,\sigma_1,\sigma_0,\delta^*)} \\
   \Leftrightarrow \underset{\substack{\sigma_j}}{\min}& \int_{\delta^*(m)=1}{{\sigma_1(m)}dm} + \lambda \int_{M}{{\sigma_1(m)\ln\frac{\sigma_j(m)}{f_j(m)}}dm}.
\label{eq:signalling_intruder2_1_2}
\end{aligned}
\end{equation}
To solve \eqref{eq:signalling_intruder2_1_2}, the concept of game potentials\cite{lloyd_1996potential_game} is helpful. A game potential $\Xi:\bar{A}_S\rightarrow \mR$ refers to a function satisfying $\Xi(\sigma_i,\sigma_{-i})-\Xi(\sigma'_j,\sigma_{-j})>0 $ if and only if $\bar{C}_S(H_j,\sigma'_i,\sigma_{-i},\delta^*) -\bar{C}_S(H_j,\sigma'_j,\sigma_{-i},\delta^*)>0$ for all $\sigma'_i,\sigma'_j$ $i,j\in \{0,1\}$.
If we can identify such a function, we can solve the game by optimizing the game potential across all possible strategy profiles of the players. 
By observation we can formulate the attacker's optimization problem as follows:
\begin{equation}
\begin{aligned}   &\underset{{(\hat{\sigma}_1,\hat{\sigma}_0)\in \hat{A}_S}}{\min}\; \int_{\hat{\delta}^*(m)=1}{\hat{\sigma}_1(m)dm} \\
    &+ \lambda \int_{m \in {M}}{{\hat{\sigma}_1(m)\ln\frac{\hat{\sigma}_1(m)}{f_1(m)}}dm} + \lambda \int_{m \in {M}}{{\hat{\sigma}_0(m)\ln\frac{\hat{\sigma}_0(m)}{f_0(m)}}dm}.
    \label{eq:stackelberg_attacker_potential}
\end{aligned}
\end{equation}

\paragraph{The optimal strategy of a passive detector}
The passive detector faces the optimization problem \eqref{eq:stackelberg_attacker_potential}. Using the method of Lagrange multipliers \cite{levy2008principles_signal_detection} we can obtain the detector's optimal decision rule $\hat{\delta}^*$ as follows:
\begin{equation}
    \forall m\in M,\;\;\hat{\delta}^*(m) = \begin{cases}
    1 & \;f_1(m)> F^{-1}_L(\alpha)f_0(m), \\
    r\in [0,1] & f_1(m)= F^{-1}_L(\alpha)f_0(m), \\
    0 & \mbox{otherwise},
    \end{cases}
    \label{eq:receiver_naive_np}
\end{equation}
where $F_{L}(\cdot)$ refers to the cumulative probability density function of the ratio of the random variable obeying distribution $f_0$, and, correspondingly, $F^{-1}_{L}(\alpha)$ is the quantile function at the probability $\alpha$.  
Based on the detector's decisions \eqref{eq:receiver_naive_np} upon every $m\in M$, we can partition the message space $M$ into region of rejection, region of acceptance, and region of uncertainty as follows: 
\begin{align}
    \hat{M}_1 &= \{m:\;f_1(m)>F^{-1}_L(\alpha)f_0(m)\}= \{m:\hat{\delta}^*(m)=1\},
    \label{def:M1} \\
     \hat{M}_0 &= \{m:\;f_1(m)<F^{-1}_L(\alpha)f_0(m)\} = \{m:\hat{\delta}^*(m)=0\},
    \label{def:M0} \\
     \hat{M}_* &= \{m:\;f_1(m) =F^{-1}_L(\alpha)f_0(m)\} = \{m:\hat{\delta}^*(m)\in [0,1]\}.\label{def:M_}
\end{align}

\paragraph{The attacker's equilibrium strategies}
We now obtain attacker's equilibrium strategies based on the detector's optimal decision $\hat{\delta}^*$ by using first-order condition on \eqref{eq:stackelberg_attacker_potential}  .
We summarize the results in the following theorem \ref{theo:stackelberg_attacker1}.  
\begin{theo}
Let $f_0,f_1$ be the distributions of the original messages in \eqref{def:hypos}.  Let $M$ be the message space specified in lemma \ref{lemma:mu_sigma_equiv} and  $\hat{M}_1, \hat{M}_0$ be two regions of messages specified in \eqref{def:M1} and \eqref{def:M0}, respectively.
Then, we can write the optimal solution $\hat{\sigma}^*_1,\hat{\sigma}^*_0$ of the attacker's optimization problem \eqref{eq:stackelberg_attacker_potential} as follows: 
\begin{align}
\hat{\sigma}^*_0(m) &= f_0(m),\;\forall m\in M,
\label{sol:stackelberg_attacker0}
    \\
\hat{\sigma}^*_1(m) & = \begin{cases}
     \frac{f_1(m)}{\int_{m\in \hat{M}_1}{e^{-\frac{1}{\lambda}}f_1(m)dm} + \int_{m\in \hat{M}_0}{{f_1(m)}dm}},\;\;& m\in \hat{M}_0, \\
     \frac{f_1(m)e^{-\frac{1 }{\lambda}}}{\int_{m\in \hat{M}_1}{{e^{-\frac{1}{\lambda}}f_1(m)}dm} + \int_{m\in \hat{M}_0}{{f_1(m)}dm}},& m\in \hat{M}_1.
    \end{cases}
    \label{sol:stackelberg_attacker1}
\end{align}
\label{theo:stackelberg_attacker1}
\end{theo}
We remark that when $\lambda = 0$, the attacker's equilibrium strategies \eqref{sol:stackelberg_attacker0}\eqref{sol:stackelberg_attacker1} reduce to
\begin{align}
     &\hat{\sigma}^*_1(m) = \begin{cases}
     \frac{f_1(m)}{\int_{m\in M_0}{{f_1(m)}dm}} & m\in \hat{M}_0,\\
     0 & m\in \hat{M}_1,
     \end{cases}  \\
     & \hat{\sigma}^*_0(m) = f_0(m),\;\forall m\in M,
\end{align}
which shows that the attacker `moves' the samples in $M_1$ and `redistributes' them over the region $M_0$. The probability of detection is $0$ if  detector applies classical NP testing method;
on the other hand, when $\lambda \rightarrow \infty$, the attacker's equilibrium strategies  reduce to $\hat{\sigma}^*_1(m) = f_1(m),\;\hat{\sigma}^*_0(m) = f_0(m),\;\forall m\in M$, meaning the attacker cannot `lie' about the message due to infinitely high penalty for distorting the messages at equilibrium. From the perspective of the detector, the Stackelberg game model reduces to the classical NP hypothesis problem.




\section{The Formulation of a Proactive Detector}
\label{sec:signalling}
In this section, we assume that the detector is aware of the attacker-in-the-middle who intercepts and sends out distorted messages to her (see Figure \ref{fig:naive_detector}).  
Once the detector becomes aware that the messages have been manipulated, she can respond strategically to mitigate the reduced detection accuracy caused by the distortion. Conversely, knowing the detector is proactive, the attacker is less able to manipulate messages as extensively as they could with a passive detector. These interactions and dynamics can be effectively modeled using the signaling game framework, which we introduce below.

\subsection{The Signaling Game Formulation}
  
\begin{figure}
    \centering    \includegraphics[scale=0.28]{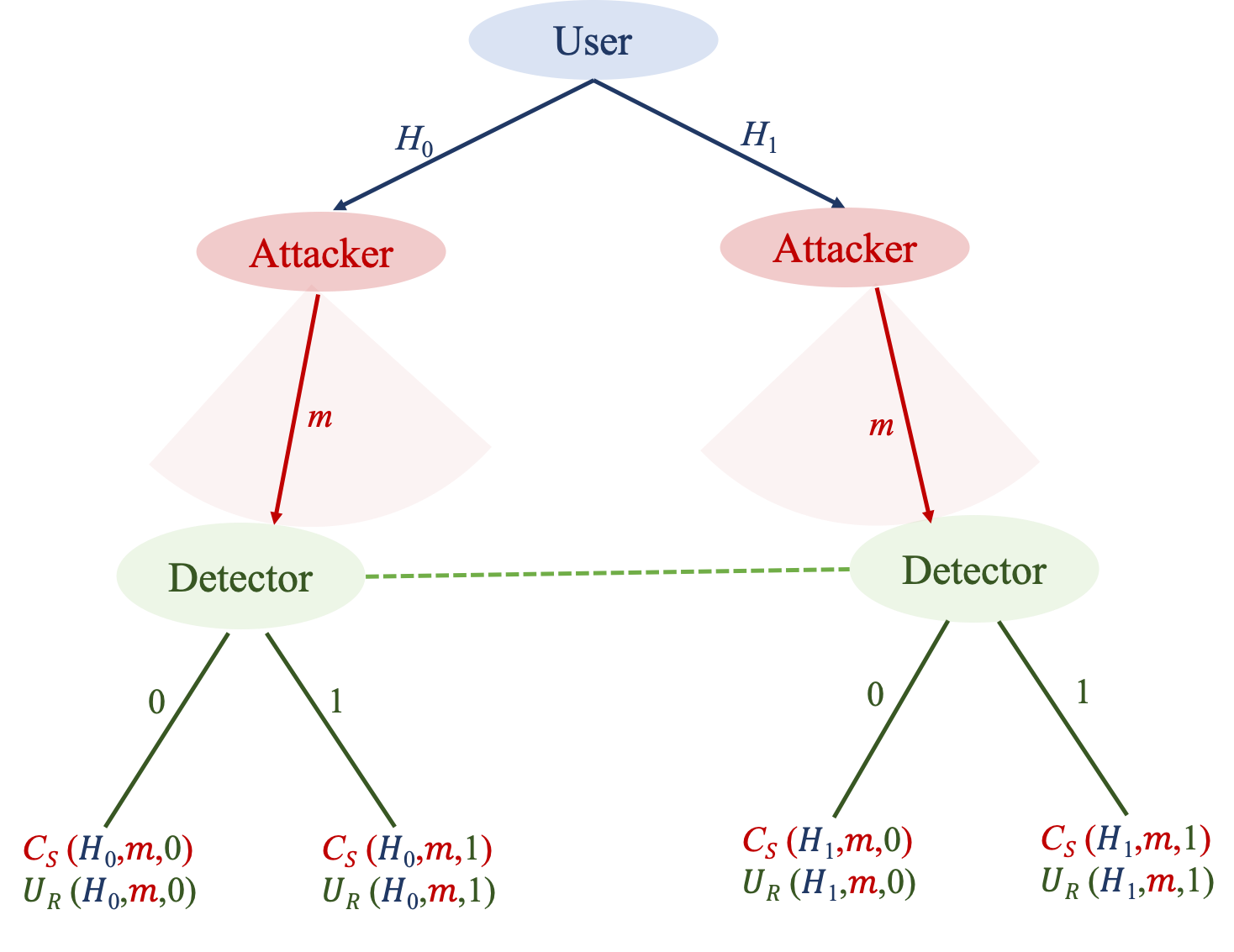}
    \caption{The game tree of the signaling game $\mc{G}_2$ capturing the interactions between the attacker and the proactive detector. The attacker observes the type (i.e., $H_0$ or $H_1$), receives a sample from the user, and sends strategically distorted message $m$ to the detector. Finally, the detector makes a decision $i\in\{0,1\}$ based on the manipulated message and the posterior belief. The payoffs for each outcome of the game are depicted at the bottom of the  game tree.}
    \label{fig:htg}
\end{figure}

We describe the gameplay between a proactive detector and an attacker. Let $p(H_0), p(H_1)$ be the common prior of the user's true type (private information).  When the game starts, the \textit{attacker} observes the user's type and receives a message $m'\in M$ from it. Based on the user's type $H_i,\;i\in\{0,1\}$, the attacker strategically generates a piece of distorted message $m = \mu_i(m')$ and sends it to the detector. 
The detector determines a detection rule ${\delta}: M\rightarrow \{0,1\}, \delta\in\Gamma,$ based on $m$ and the posterior belief, where $\Gamma$ is the set of admissible detection rules. We let $a = \delta(m),\;a \in \bar{A}_R = \{0,1\}$ be a generic action of the detector for an input $m\in M$. Finally, the detector receives a reward $U_R:\mc{H}\times M \times \bar{A}_R\rightarrow \mR$ and the attacker receives a cost $C_S:\mc{H}\times M \times \bar{A}_R\rightarrow \mR$ based on the user's true type $H_i$, the message $m$, and the detector's decision $\delta(m) = j\in\{0,1\}$. The gameplay is summarized using the game tree depicted in Figure \ref{fig:htg}. Following the similar arguments used in Section \ref{sec:stackelberg_game}, we formulate the attacker's strategy space $\bar{A}_S$ in terms of the resulting measures of the distorted messages: 
\begin{equation}
\begin{aligned}
    \bar{A}_S &= \Big\{\begin{bmatrix}
    \sigma_0 \\ \sigma_1
    \end{bmatrix}\in {L}^1(M)^{\otimes 2}:\int_{m\in M}{\sigma_i(m)dm} = 1; \\
    \;&\forall m\in M,\sigma_i(m)\geq 0, \;i\in\{0,1\}\Big\}
    \label{eq:strategy_space_sender_signaling}
    \end{aligned}
\end{equation}

Based on message $m$, the prior belief $p(H_0),p(H_1)$ of the true hypothesis is updated to be a posterior belief $p(H_j|m),\;j\in\{0,1\}$ using Bayes' rule: 
\begin{equation}
    p(H_j|m) = \frac{\sigma_j(m)p(H_j)}{\sigma_1(m)p(H_1) +\sigma_0(m)p(H_0) },
    \label{posterior_belief}
\end{equation}

\paragraph*{Detector's utility function} 
{
Upon receiving the message $m$, a proactive detector aims to maximize the estimated \textit{detection rate} while constraining the \textit{false alarm rate} to $\alpha$, where $\alpha$ is a threshold parameter determined by the detector. To effectively utilize the message $m$, the proactive detector resorts to two conditional probabilities $p(H_1|m),p(H_0|m)$-posterior beliefs presented in \eqref{posterior_belief}-to estimate the detection rate. As shown in the game tree in Figure \ref{fig:htg}, the proactive detector knows that the user is {normal} with probability $p(H_0|m)$ (left branch of the game tree) and  {abnormal} with probability $p(H_1|m)$ (right branch of the game tree). Furthermore, if the detector decides $\delta(m) = 1$, she successfully detects an abnormal user but incurs a false alarm if the user is normal. Conversely, if she chooses $\delta(m)=0$, her decision results in neither detection nor false alarm, regardless of the user's type.  For a given message $m$, the detection rate is $\delta(m)p(H_1|m)$, and the false alarm rate is $\delta(m)p(H_0|m)$. The detector must carefully weigh the risk of incurring a false alarm against the potential gain in detection accuracy when opting for $\delta(m)=1$.} 
Thus, we formulate the detector's optimization problem as follows:
\begin{equation}
\begin{aligned}
    \underset{\delta\in\Gamma}{\sup}&\;\delta(m){p(H_1|m)},\;\;\text{s.t.}&\; \delta(m) p(H_0|m) \leq \alpha.
     \label{opt:proactive_detector}
\end{aligned}
\end{equation} 
In other words, we can designate the detector's utility as the objective function in \eqref{opt:proactive_detector} as follows: 
\begin{equation}
    \bar{U}_R(m,\delta(m)) = \sum_{i\in \{0,1\}}{p(H_i|m)U_R(H_i,m, \delta(m))},
    \label{eq:util_detector_proactive}
\end{equation}
 with $U_R(H_1,m,1)=1,\;U_R(H_1,m,0) = 0,\;U_R(H_0,m,0) = 0,\;U_R(H_0,m,1) = 0$.
and the detector's strategy space is chosen as 
\begin{equation}
    \bar{\Gamma} = \{\delta\in 2^M:\;\delta(m)p(H_0|m)\leq \alpha\}.
    \label{def:detector_signaling_strategy}
\end{equation}

\paragraph*{Attacker's cost function}

Similar to Stackelberg game scenario as introduced in Section \ref{sec:stackelberg_game}, the attacker faces the trade-off between undermining the detector's performance and lowering the costs due to distortion of messages. In our framework, the attacker knows the model including the distributions $f_0,f_1$ as well as the detector's utility function. The assumption of an omnipotent attacker aligns with the Kerckhoff's principle \cite{shannon1949kerckhoff} used for adversarial analysis. We introduce \textit{de facto} detection rate $\bar{P}_D: \bar{\Gamma} \rightarrow [0,1]$ and false alarm rate $\bar{P}_F: \bar{\Gamma} \rightarrow [0,1]$:
\begin{align}
     \bar{P}_D(\delta^*) &=  \int_{{\delta}^*(m)=1}{\sigma_1(m)dm},\;
     \label{detection_rate_signalling}\\ \bar{P}_F(\delta^*) &=  \int_{{\delta}^*(m)=1}{\sigma_0(m)dm}.
  \label{false_alarm_rate_signalling}
\end{align} 
The detector makes decisions based on the attacker's message whose distribution are distorted from $f_0$ to $\sigma_0$ and from $f_1$ to $\sigma_1$. Notice that the attacker does not modify the `label' or the index of the hypothesis (i.e., $i\in\{0,1\}$) of the message. The \textit{de facto} detection rate measures the consequence of the detection rate under distorted messages. From the attacker's viewpoint, the detector knows his existence and does her best to distinguish between $H_0$ and $H_1$ from the received messages (from the man-in-the-middle). The attacker anticipates it and aims to minimize the \textit{de facto} detection rate to reduce the accuracy for the detector to tell $H_1$ or $H_0$ from the distorted messages. Thus the quantity \eqref{detection_rate_signalling} is included in the attacker's cost function. 

Considering the KL divergence as the penalty function,  we formulate the attacker's optimization problems $j\in\{0,1\}$ as follows:
    \begin{equation}
\begin{aligned}
 \underset{\substack{\sigma_j\in \bar{A}_S}}{\min}& \int_{\delta^*(m)=1}{{\sigma_1(m)}dm} + \lambda \int_{m \in M}{{\sigma_j(m)\ln\frac{\sigma_j(m)}{f_j(m)}}dm}.
    \label{eq:signalling_intruder2_1}
\end{aligned}
\end{equation}

    Notice there is no genuine false alarm rate included in attacker's cost function \eqref{eq:signalling_intruder2_1}, implying that the attacker values causing mis-detection over false alarms. We propose framework of proactive detector using a signaling game as in the following definition.

\begin{Def}[The game of a proactive detector]
We model the relationship between the proactive detector and the attacker-in-the-middle, shown in Figure \ref{fig:naive_detector}), as a signaling game $$\mc{G}_2 = \langle \mc{I}, \mc{H}, p, M,\bar{\Gamma},\bar{A}_R, \bar{A}_S, \bar{U}_R , \bar{C}_S\rangle,$$ where $\mc{I} = \{\text{Attacker},\text{Detector}\}$ is the set of players (i.e., attacker is the sender (see the player in red in Figure \ref{fig:naive_detector}), detector is the receiver(see the player in green in Figure \ref{fig:naive_detector})); $\mc{H} = \{H_0,H_1\}$ is the user's type space (see the left branch in Figure \ref{fig:naive_detector}) with $p =[p(H_0)\;p(H_1)]$ as the prior of the attacker's type; $M$ is the message space in lemma \ref{lemma:mu_sigma_equiv};  \;$\bar{\Gamma}$ represents the detector's strategy space in \eqref{def:detector_signaling_strategy} and $\bar{A}_R = \{0,1\}$ refers to the detector's action space; $\bar{A}_S$ refers to the attacker's action space defined in \eqref{eq:strategy_space_sender_signaling}; $\bar{U}_R: M \times \bar{A}_R\rightarrow \mR$ is the utility function of the detector defined in \eqref{eq:util_detector_proactive}; $\bar{C}_S:\mc{H}\times M \times \bar{\Gamma} \rightarrow \mR$ is the cost function of the attacker defined in \eqref{eq:signalling_intruder2_1}.
\label{Def:signalling_game}
\end{Def}

\textbf{Remark:} The proposed game model in definition \ref{Def:signalling_game} differs from previous game-theoretic models for network detection in that we adopt a signaling game, where the detector (receiver) can update her belief and thus update her detection rule based on attacker's message, characterizing the concept of `proactive defense.' Compared to detection models based on machine learning techniques, the proposed framework is not entirely data-driven. It thus does not require so great amount of data to reach a precise performance as the machine-learning-based counterpart. We will add parameters from real datasets, as discussed in later sections.

\subsection{Equilibrium Analysis}
To solve the signaling game $\mc{G}_2$, we adopt the concept of Perfect Bayesian Nash Equilibrium (PBNE)\cite{fudenberg1998game},  {which aims at solving the optimization problem \eqref{eq:receiver_np}\eqref{eq:signalling_intruder2_1} under belief consistency condition \eqref{posterior_belief}. It is clear that only hybrid PBE exist. }

\paragraph{Detector's Optimal Decision Rule}

For a given message $m$, the proactive detector faces the constrained optimization problem \eqref{eq:receiver_np}.  We obtain the detector's optimal strategy explicitly $\delta^*$ by solving \eqref{eq:receiver_np} as follows:  
\begin{equation}
    \delta^*(m) = \begin{cases}
    1 & p(H_0|m)<\alpha, \\
    r\in (0,1) & p(H_0|m)=\alpha, \\
    0 & p(H_0|m)>\alpha.
    \end{cases}
    \label{eq:receiver_np}
\end{equation}
 or equivalently $\delta^*(m)=1$ if $\sigma^*_1(m)>\beta\sigma^*_0(m)$ and $\delta^*(m)=0$ otherwise, where $\beta = \frac{p(H_0)}{p(H_1)}\left(\frac{1}{\alpha}-1\right)$ is the threshold of the likelihood ratio test. 

\paragraph{Attacker's Equilibrium Strategies} 

By sequential rationality and the expression \eqref{eq:signalling_intruder2_1}, the attacker's equilibrium strategies $\sigma^*_1,\sigma^*_0$ are best-responses of the detector's optimal decision rule $\delta^*$ in \eqref{eq:receiver_np}. 
As seen in the Stackelberg game case in Section \ref{sec:stackelberg_game}, with the help of exact game potentials \cite{lloyd_1996potential_game}, we manage to convert the attacker's optimization problems \eqref{eq:signalling_intruder2_1} into one with a single decision-maker:
\begin{equation}
\begin{aligned}
    &\underset{\substack{(\sigma_1,\sigma_0)\in \bar{A}_S}}{\min}\; \int_{\delta^*(m)=1}{\sigma_1(m)dm} \\
    &+ \lambda  \int_{m\in M}{\left({\sigma_1(m)\ln\frac{\sigma_1(m)}{f_1(m)}} + \sigma_0(m)\ln\frac{\sigma_0(m)}{f_0(m)}\right)dm}.
\label{eq:potential}
\end{aligned}
\end{equation}

    To clarify our game-theoretic framework in Definition \ref{Def:signalling_game}, we introduce the assumption on the support of the two underlying distributions $f_0,f_1$, which refers to the largest connected set upon which the value of function is nonzero. 
\begin{assume}[Full-overlapping support]
Let $f_0,f_1\in L^1(M)$ be the distributions of the original messages from the user mentioned in \eqref{def:hypos}. We assume that the two distributions have the same support, i.e., $\supp f_1  = \supp f_0$. 
\label{assume:full_support}
\end{assume}
Assumption \ref{assume:full_support} guarantees that the non-adversarial detection problem (i.e., the classical hypothesis testing problem) is non-trivial. If we drop this assumption, then all messages lying outside their overlapping area will reveal the true hypothesis immediately.   

\begin{assume}[Absolute continuity]
\label{assume:abs_cts}
Let $\Sigma_1,\Sigma_0$ be the measures of distorted messages and $F_0,F_1$ be the measures of original messages in Lemma \ref{lemma:mu_sigma_equiv}. Suppose that $\sigma_1,\sigma_0, f_1,f_0\in L^1(M)$ as in \eqref{eq:potential} are Radon-Nikodym derivatives of $\Sigma_1,\Sigma_0, F_1,F_0$  with respect to the Lebesgue measure on $M$. Then, we assume $\Sigma_1\ll F_1,\; \Sigma_0\ll F_0$ when $\lambda>0$ in the objective function in \eqref{eq:potential}.
\end{assume}
Assumption \ref{assume:abs_cts} guarantees that the KL divergence terms in the attacker's objective functions are well-defined. We arrive at the following lemma.
\begin{lemma}
Let $\sigma_1,\sigma_0, f_1,f_0$ be the distributions mentioned in Assumption \ref{assume:abs_cts}. Then, we have 
\begin{equation}
    \supp \sigma_1 = \supp f_1,\;\;\supp \sigma_0 = \supp f_0,
    \label{eq:abs_cts_same_support}
\end{equation}
and furthermore,
\begin{equation}
    \supp \sigma_1 = \supp\sigma_0.
    \label{eq:sig1_sig0_same_support}
\end{equation}
\label{lemma:full_support}
\end{lemma}
 The conclusion \eqref{eq:sig1_sig0_same_support} is the result of combining Lemma \ref{lemma:full_support} and Assumption \ref{assume:abs_cts}. We categorize the solutions to \eqref{eq:potential} in terms of three regimes: $\lambda = 0,\;\lambda=\infty$,  and $\lambda\in (0,\infty)$. 

\subsubsection{Detection Rate Only: When $\lambda = 0$}
When $\lambda = 0$, the attacker does not suffer a penalty from distorting the messages, and the signaling game $\mc{G}_2$ can be considered as a `cheap-talk' game \cite{crawford1982cheap_talk}.  We arrive at the following conclusions on the attacker's equilibrium strategies. 
\begin{prop}
\label{prop:sol_signaling_lam_0}
Let $\mc{G}_2$ be the signaling game defined in Definition \ref{Def:signalling_game}. Let \eqref{eq:potential} be the optimization problem for the attacker with $\lambda = 0$. 
Then, the attacker's strategies at equilibrium are all such that the following constraints are met.
\begin{enumerate}
    \item  When $0<\beta<1$, the attacker's equilibrium strategies $\sigma^*_1,\sigma^*_0$ satisfy the following constraints:
\begin{align}
    \sigma^*_1(m)&= \begin{cases}
    \beta \sigma^*_0(m) &  m\in \supp \sigma^*_0, \\
    (1-\beta)\varphi(m) & \;m \notin \supp \sigma^*_0,\label{eq:sig_s_h1}
    \end{cases} \\
     \text{with}&\;\int_{m \notin \supp \sigma^*_0}{\varphi(m)dm} = 1.
\end{align}     
The regions of rejection and of acceptance are:
\begin{equation}
    M_0 = \supp \sigma^*_0,\;\; M_1 = \supp \varphi;
    \label{eq:sig_s_h0}
\end{equation}

\item When $\beta \geq 1$,  distributions that satisfy the constraint
\begin{equation}
   \forall m\in M,\;\sigma^*_1(m) \leq \beta \sigma^*_0(m),\;\;
\end{equation}
 are equilibrium strategies for the attacker. 
The corresponding regions of rejection and of acceptance are: $M_1 = \varnothing,\;M_0 = M$.
\end{enumerate}
\end{prop}
For the proof of the proposition, we refer to Section \ref{subsec:lam_0_sender} in the appendix. 

\subsubsection{The KL Divergence Only:  $\lambda\rightarrow\infty$}
When $\lambda\rightarrow\infty$ in \eqref{eq:potential}, the attacker's cost function is completely measured by the KL divergence of the measures of distorted messages to the ones of original messages. We show in the following that the detector's best response toward $\delta^*$ is to adopt a `disclosure strategy': i.e., the attacker delivers messages to the detector without distortion. 

\begin{prop}
\label{prop:sol_signaling_lam_infty}
Let $\mc{G}_2$ be the signaling game defined in Definition \ref{Def:signalling_game}, and let \eqref{eq:potential} be the optimization problem for the attacker with $\lambda\rightarrow\infty$. Then, the attacker's strategies are as follows:
\begin{equation}
   \forall m\in M,\;\; \sigma^*_1(m) = f_1(m),\;\; \sigma^*_0(m) = f_0(m),
    \label{eq: kl_divgence_min}
\end{equation}
and the regions of rejection and of acceptance depend on $f_0(m),f_1(m)$: $M_1 = \left\{m: f_1(m)>{\beta}f_0(m)\right\},\;M_*  = \left\{m: f_1(m)={\beta}f_0(m)\right\}, \;M_0 = \left\{m: f_1(m)<{\beta}f_0(m)\right\}.$
\end{prop}
For the proof of the proposition, we refer to Section \ref{subsec:lam_infty_sender} in the appendix. 

\subsubsection{General Case: $0<\lambda <\infty$}
In the general case $\lambda\in (0,\infty)$, the attacker faces the trade-off between undermining the detector's performance and minimizing the penalty resulting from message distortion. We categorize PBNE into three classes: separating, pooling, and hybrid equilibrium. It is not difficult to find out there do not exist equilibria of the first two kinds. Thus we focus on deriving the hybrid equilibrium.

To obtain the attacker's best-response strategies $\sigma^*_1,\sigma^*_0$ under the detector's optimal decision rule $\delta^*$ as in \eqref{eq:receiver_np}, we solve the generalized KKT conditions \cite{luenberger1997optimization} of the optimization problem \eqref{eq:potential} and arrive at the following sufficient conditions for finding a pair of equilibrium strategies for the attacker:

\begin{assume}
\label{para}
Let \eqref{eq:potential} be the optimization problem for the attacker in the signaling game $\mc{G}_2$ with $0<\lambda<\infty$. We assume that there exists $\zeta\in \left[\underset{m\in M}{\sup} \frac{f_0(m)}{f_1(m)},\;e^{\frac{1}{\lambda}}\underset{m\in M}{\sup} \frac{f_0(m)}{f_1(m)}\right]$ such that 
     \begin{equation}
          J(\zeta) =  \int_{\frac{f_1(m)}{f_0(m)}<\frac{1}{\zeta}}{f_0(m)dm} + \int_{\frac{f_1(m)}{f_0(m)}>\frac{1}{\zeta}e^{\frac{1}{\lambda}}}{f_0(m)dm} ,
          \label{zeta_beta_1}
     \end{equation}
     where $J: \mR_+ \rightarrow \mR$ is a function defined as
    \begin{equation}
\begin{aligned}
    &J(\zeta)  =  \beta\zeta \left(\int_{\frac{f_1(m)}{f_0(m)}<\frac{1}{\zeta}}{f_1(m)dm} + \int_{\frac{f_1(m)}{f_0(m)}>\frac{1}{\zeta}e^{\frac{1}{\lambda}}}{f_1(m)e^{-\frac{1}{\lambda}}dm}\right)\\
    &+(\beta-1)\zeta^{\frac{\beta}{1+\beta}}\left(\int_{\frac{1}{\zeta}\leq \frac{f_1(m)}{f_0(m)}\leq\frac{1}{\zeta}e^{\frac{1}{\lambda}}}{f_1(m)\left(\frac{f_1(m)}{f_0(m)}\right)^{-\frac{1}{1+\beta}}dm}\right).
    \label{zeta_beta_2}
\end{aligned}
\end{equation}
\label{lemma:function_J_zeta}
\end{assume}
We discuss the proof of the lemma in Section \ref{sec:sender_optimal_strategy_finite_lambda} in the appendix. Now we argue through the following proposition that both the attacker's strategies at equilibrium and the detector's region of rejection can be parameterized by $\zeta$.
\begin{prop}
Let $\mc{G}_2$ be the signaling game defined in Definition \ref{Def:signalling_game}. The attacker's strategies at equilibrium $\sigma^*_1,\sigma^*_0$ can be categorized in terms of $\lambda$:
 When $0<\lambda<\infty$, the attacker admits the following strategies at equilibrium:
\begin{align}
    \sigma^*_0(m) &=\begin{cases}
    c_0f_0(m)  & m\in M_0, \\
    c_0 f_1^{\frac{\beta}{1+\beta}}(m)f_0(m)^{\frac{1}{1+\beta}}(m)\zeta^{\frac{\beta}{1+\beta}}  & m\in M_*,\\
    c_0f_0(m) & m\in M_1;
    \end{cases} \label{eq:signal_equil_0} \\
    \sigma^*_1(m) & = \begin{cases}
    c_1f_1(m) & m\in M_0,\\
    c_1f_1^{\frac{\beta}{1+\beta}}(m)f_0(m)^{\frac{1}{1+\beta}}(m)\zeta^{-\frac{1}{1+\beta}} & m\in M_*, \\
    c_1f_1(m) e^{-\frac{1}{\lambda}} & m\in M_1;
    \end{cases}
\label{eq:signal_equil_1}
\end{align}
with the normalization factors $c_0,c_1$ as
\begin{align}
    {1}/{c_0} & = \int_{M_0}{f_0(m)dm}+\int _{M_*}{f_0(m)\left(\frac{f_1(m)}{f_0(m)}\right)^{\frac{\beta}{1+\beta}}\zeta^{\frac{\beta}{1+\beta}}dm} \\
    &+\int_{M_1}{f_0(m)dm},
    \label{c0_relabel}\\
    {1}/{c_1}&= \int_{M_0}{f_1(m)dm}+\int_{M_*}{f_1(m)\left(\frac{f_1(m)}{f_0(m)}\right)^{-\frac{1}{1+\beta}}\zeta^{-\frac{1}{1+\beta}}dm} \\
    &+\int_{M_1}{f_1(m)e^{-\frac{1}{\lambda}}dm}
      \label{c1_relabel}.
\end{align} 
The exact regions $M_0,M_*,M_1$ (strict acceptance region, uncertainty region, strict rejection region) can be characterized in terms of $\zeta$:
\begin{align}
    M_0 &= \left\{m: \frac{f_1(m)}{f_0(m)}<\frac{1}{\zeta}\right\},\;\;\\
    M_* &= \left\{m:\frac{1}{\zeta}\leq \frac{f_1(m)}{f_0(m)}\leq e^{\frac{1}{\lambda}}\frac{1}{\zeta}\right\}, \\
    M_1 &= \left\{m: \frac{f_1(m)}{f_0(m)}>e^{\frac{1}{\lambda}}\frac{1}{\zeta}\right\}.
\end{align}
\label{prop:finite_lam_signaling}
\end{prop}
For brief derivations of the proposition, we refer to Section \ref{sec:sender_optimal_strategy_finite_lambda} in the appendix. 
    The complete proof is provided in the supplementary document associated with this paper. In Figure \ref{fig:region of rejection}, we illustrate the three regions to provide a visual understanding. We also include an example where $f_1,f_0$ are both binomial distributions, illustrating how the attacker's manipulation upon the vanilla distributions affects the region of rejection against a proactive detector.   
  \begin{figure}     
    \centering
      \includegraphics[width=0.5\textwidth]{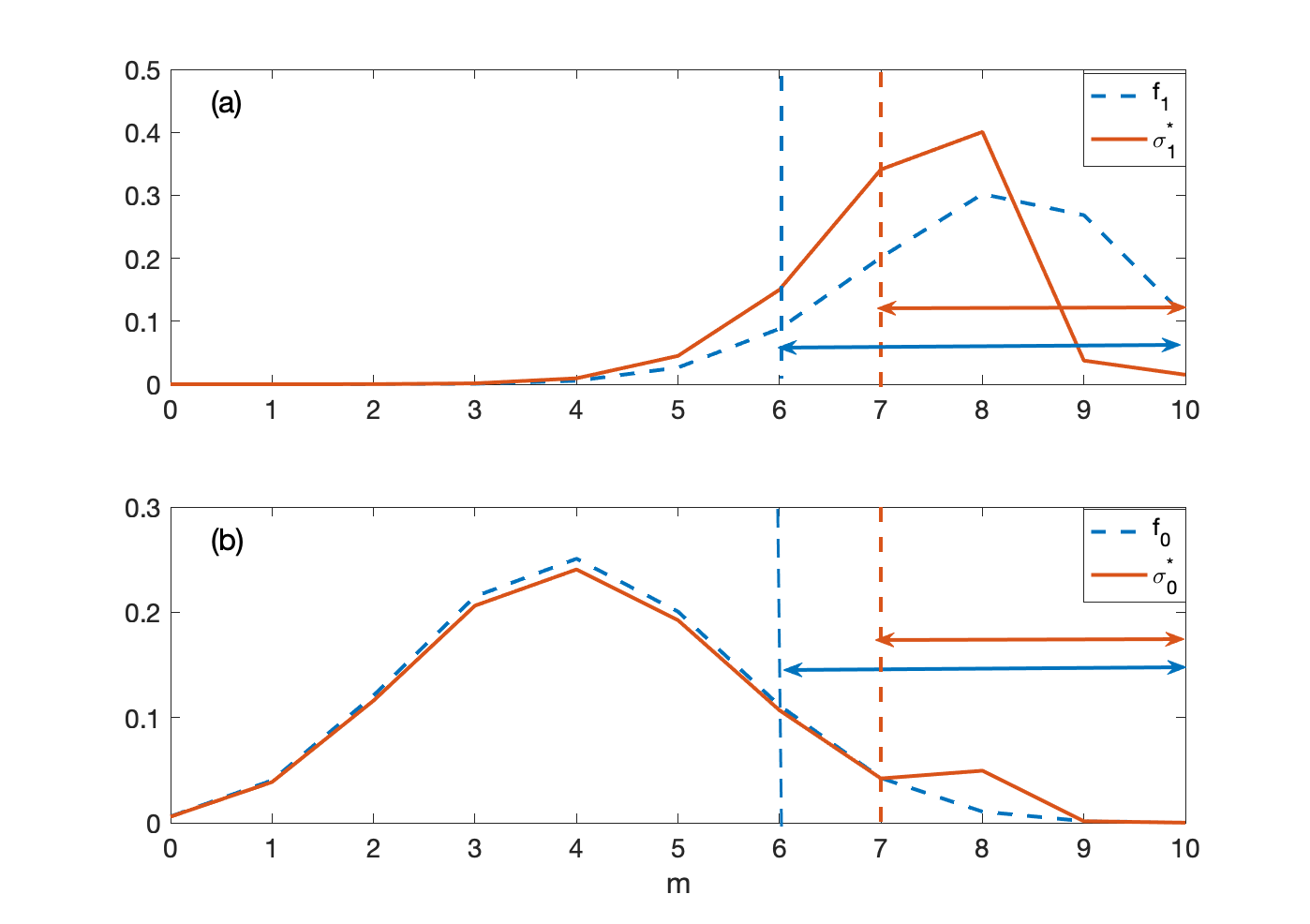} 
        \caption{{An illustration of attacker's optimal strategies $\sigma^*_1,\sigma^*_0$ (red, dashed lines) in comparison with the original distributions $f_1,f_0$ (blue, solid lines). The vanilla distributions $f_1,f_0$ are binom$(10,0.8)$ and binom$(10,0.4)$. We choose $\lambda=0.4$. The red arrows indicate the region of rejection under the attacker's manipulation $M_1$, while the blue arrows indicate the uncorrupted region of rejection $\hat{M}_1$. }}
        \label{fig:region of rejection}
    \end{figure}
\paragraph*{Remark}
It is also possible to use the \textit{counterfactual} detection rate $P^*_D: \Gamma^*:\rightarrow [0,1]$
\begin{equation}
    P^*_D(\delta^*) = \int_{\delta^*(m)=1}{f_1(m)dm},
    \label{pd_counterfactual}
\end{equation}
as the first term in the cost function of the attacker. The quantity in \eqref{pd_counterfactual} measures the detection performance when detector uses the optimal detection rule yet under messages directly from the nature. The attacker’s goal is to directly minimize the counterfactual detection rate. By observing \eqref{pd_counterfactual}, the attacker can be viewed as someone who minimizes the weighted size of the detection region. 
Using \eqref{pd_counterfactual} in the attacker’s utility captures a non-myopic attacker who cares not only his direct interaction with the detection but the holistic sequential interaction (from Nature to Attacker to Detector). However, one flaw is that detection rate \eqref{pd_counterfactual} is an imaginary, conjectural, or counterfactual quantity that is not directly measurable (or observable) by the detector himself. The following proposition characterizes the relationship between counterfactual and \textit{de facto} detection rates.
\begin{prop}
Let $f_0,f_1,\sigma_1,\sigma_0\in L^1(M)$ be as  in assumption \ref{assume:abs_cts}. Let $P_D,\bar{P}_D:\bar{\Gamma}\rightarrow [0,1]$ be two detection rates given in \eqref{detection_rate_signalling} and \eqref{pd_counterfactual}.
Then $P_D(\delta^*)\geq {P}^*_D(\delta^*)$ when the attacker and the detector choose equilibrium strategies. 
\label{prop:defacto_counterfacto}
\end{prop}
\begin{proof}
It follows from \eqref{c1_relabel} that $e^{-\frac{1}{\lambda}}\leq c_1\leq 1$ so the attacker's equilibrium strategy $\sigma^*_1 = c_1f_1(m)e^{-\frac{1}{\lambda}}\leq f_1(m)$ for all $m\in M_1$, which proves the proposition. 
\end{proof}

    \textbf{Remark:} The previous discussions on equilibrium assume that $ p(H_0)\sigma^*_0(m) + p(H_1)\sigma^*_1(m) > 0$, suggesting the detector could only receive messages lying in the support of the attacker's equilibrium strategies. It is possible that this assumption is not met, but we can show that such an `off-equilibrium' path does not constitute a PBNE for the signaling game.

\section{Neyman-Pearson Detectors under repeated observations}
\label{sec:repeat_obs}
In this section, we discuss the performance of proactive detectors under repeated observations. 
We assume that the user's type stays the same throughout the game, i.e., we adopt an \textit{open-loop} information structure \cite{bacsar1998dynamic_game}. 

    A prior belief $p(H_0),p(H_1)$ is assigned to the attacker and the detector as common knowledge, respectively. At the beginning of each stage, the user sends identically distributed messages following the distributions $f_0$ or $f_1$ depending on their true type, as described in \eqref{def:hypos}. The attacker manipulates the message based on his observation of the user's type and his mixed strategies from previous stages. After the detector receives the distorted messages, she infers the user's true type at the end of the stage. 

\begin{Def}[History of action profiles]
We define the history of action profiles up to stage $N_1$, denoted as $h^{(j)}\in M^{\otimes j}\times \{0,1\}^{\otimes j},\;j\in[N_1]$, as follows:
\begin{equation}
    h^{(j)} = (m^{(j)}, a^{(j)}),
\end{equation}
where $m^{(j)} = (m_1,\dots,m_{j})\in M^{\otimes j}$ as a generic history of messages up to stage $j$ and  $a^{(j)}\in \{0,1\}^{\otimes j} = A^{\otimes j}_R$ refers to history of the detector's actions up to stage $j$. 
\end{Def}
 In general, at the beginning of every stage $j$, the attacker's mixed strategy and the detector's optimal decision rule should depend on the history $h^{(j-1)}$.
 The set of mixed strategies available to the attacker, denoted as $\bar{A}^{(j)}_S$ for $j\in [N_1]$, is expressed as follows
\begin{equation}
\begin{aligned}
  \bar{A}^{(j)}_S &= \Big\{\begin{bmatrix}
    \sigma^{(j)}_0 \\ \sigma^{(j)}_1
    \end{bmatrix}\in ({L}^1(M))^{\otimes 2}: \\
    &\int_{m'\in M}{\sigma^{(j)}_i(m')dm'} = 1;\;\forall m'\in M,\sigma^{(j)}_i(m)\geq 0, \;i\in\{0,1\}\Big\},
\end{aligned}
    \label{eq:strategy_space_sender_bar}
\end{equation}
where the proactive detector obtains a posterior belief $p(H_1|m^{(j)})$ based on history of distorted messages $m_1,\dots m_{j-1}$ and the current message $j$ as expressed in \eqref{belief_system}. Applying the equation \eqref{belief_system} recursively for $j=1,2,\dots N_1$ we get the following:
\begin{equation}
    p(H_1|m^{(j)}) = \frac{\prod_{k=1}^{j}{\sigma^{(k)*}_1(m_k)}p(H_1)}{\prod_{k=1}^{j}{\sigma^{(k)*}_1(m_k)}p(H_1) + \prod_{k=1}^{j}{\sigma^{(k)*}_0(m_k)}p(H_0)}
    \label{eq:belief_system_equil_j}
\end{equation}
   In our multi-stage game, both the attacker and the detector know their actions in previous stages, or have perfect recalls, so by Kuhn's theorem \cite{aumann1961kuhn_theorem} it suffices to only consider the mixed strategies of attackers and behavior strategies of detectors (or detection rules) throughout the game. 

In general, a dynamic Bayesian game could yield many equilibrium concepts. We state in the following proposition on the most natural one extended from 2-shot signaling game $\mc{G}_1$.
\begin{prop}
\label{prop:s_pbne}
    Let the multi-stage game be the one extended from $\mc{G}_1$. Then under mild technical assumptions, we can write the optimal decision rule explicitly as follows:
\begin{equation}
     \bar{\delta}^{*(j)}(m_{j}|m^{(j-1)}) = \begin{cases}
    1 & {\frac{\sigma^{(j)*}_1(m_j)}{\sigma^{(j)*}_0(m_j)}} > \beta_{j}, \\
    r\in (0,1) & {\frac{\sigma^{(j)*}_1(m_j)}{\sigma^{(j)*}_0(m_j)}}= \beta_{j},  \\
    0 & {\frac{\sigma^{(j)*}_1(m_j)}{\sigma^{(j)*}_0(m_j)}} < \beta_{j},
    \end{cases}
    \label{eq:delta_stage_j1}
\end{equation}
with 
\begin{equation}
    \beta_j =   \prod_{k=1}^{j-1}{\frac{\sigma^{(k)*}_0(m_k)}{\sigma^{(k)*}_1(m_k)}}\beta^j,\;\;\text{and}\;\beta^j = \frac{p(H_0)}{p(H_1)}\left(\frac{1}{\alpha_j}-1\right),     \label{eq:beta_j}
\end{equation}
where $\beta$ is the same threshold as in \eqref{eq:receiver_np}. 
We now study the attacker's behaviors. 

In the meantime, we thus obtain the attacker's equilibrium strategies at stage $j$ as 
\begin{align}
    \sigma^{(j)*}_0(m) &=\begin{cases}
    c^j_0f_0(m)  & m\in M^j_0, \\
    c^j_0 f_1^{\frac{\beta_j}{1+\beta_j}}(m)f_0^{\frac{1}{1+\beta_j}}(m)\zeta_j^{\frac{\beta_j}{1+\beta_j}}  & m\in M^j_*,\\
    c^j_0f_0(m) & m\in M^j_1;
    \end{cases} \label{eq:signal_equil_0_j} \\
    \sigma^{(j)*}_1(m) & = \begin{cases}
    c^j_1f_1(m) & m\in M^j_0,\\
    c^j_1f_1^{\frac{\beta_j}{1+\beta_j}}(m)f_0^{\frac{1}{1+\beta_j}}(m)\zeta_j^{-\frac{1}{1+\beta_j}} & m\in M^j_*, \\
    c^j_1f_1(m) e^{-\frac{1}{\lambda}} & m\in M^j_1;
    \end{cases}
\label{eq:signal_equil_1_j}
\end{align}
with the normalization factors guaranteeing that $\int_{M}{\sigma^{(j)}_0(m)dm}=1$ and $\int_{M}{\sigma^{(j)}_1(m)dm}=1$. The exact regions  $M^j_0,M^j_*,M^j_1$ (strict acceptance region, uncertainty region, strict rejection region) can be characterized in terms of $\zeta_j$:
$ M^j_0 = \left\{m: \frac{f_1(m)}{f_0(m)}<\frac{1}{\zeta_j}\right\},\;\;M^j_* = \left\{m:\frac{1}{\zeta_j}\leq \frac{f_1(m)}{f_0(m)}\leq e^{\frac{1}{\lambda}}\frac{1}{\zeta_j}\right\},\;M^j_1 = \left\{m: \frac{f_1(m)}{f_0(m)}>e^{\frac{1}{\lambda}}\frac{1}{\zeta_j}\right\}$.
\end{prop}

We leave the discussion on technical details in the appendix. 
\section{Case Study: Intrusion Detection Evasion}
\label{sec: numerical_experiment}
In this section, we conduct numerical experiments to evaluate the performance of evasion-aware detectors against intrusion detection evasion techniques \cite{cheng2011evasion} such as IMAP, backdoor denial-of-service (DoS), etc, that camouflage the true state of users.  {Intrusion detection problem enjoyed a long history of attention. Authors in \cite{sommer2010ML_intrusion_detection} examine the differences of implementing machine learning techniques in the intrusion detection domain and in other areas. Authors in \cite{durkota2017detecting_data_exfiltration} seek a sequential game framework to explore the differences between insider and outsider attackers in data exfiltration and propose algorithms approximating optimal countermeasures. Authors in \cite{lisy2014randomized_operating_point_adversarial} adopt a Nash game framework to argue an optimal point on ROC curve against a rational attacker in the intrusion detection problem. Authors in \cite{laszka2016threshold_IDS} aim at finding the optimal threshold for multiple intrusion detection systems against strategic attacks. Authors in \cite{dritsoula2017game_adversarial_classifiction} formulate an adversarial classification scenario where an attacker and a defender play a nonzero-sum game for data samples of class 1, while data of class 0 are random and exogenous and implement the frameworks on signature-based intrusion detection problem. In this paper, we demonstrate how a proactive intrusion detection system can identify the abnormalities of users under packets whose features are strategically modified by attackers to evade detection. }

\subsection{The Dataset}
 {We select a portion of the KDD Cup 1999 dataset that was built as part of the 1998 DARPA Intrusion Detection Evaluation Program, which was prepared and managed by MIT Lincoln Labs. The objective was to survey and evaluate research in intrusion detection. The 1999 KDD intrusion detection contest \cite{misc_kdd_cup_1999_data_130} uses a version of this dataset. Lincoln Labs set up an environment to acquire nine weeks of raw TCP dump data for a local-area network (LAN) simulating a typical U.S. Air Force LAN. They operated the LAN as if it were a true Air Force environment, but peppered it with multiple attacks, the main four categories of which are denial-of-service, unauthorized remote access, unauthorized access from a local superuser, and probing (such as port scan).  The raw (binary) data was about 4 gigabytes of compressed binary TCP dump data from 7 weeks of network traffic, which was later compressed into about five million connection records. A connection is a sequence of TCP packets starting and ending at some well-defined times, between which data flows to and from a source IP address to a target IP address under some well-defined protocol. Each connection record, consisting of 100 bytes, is labeled as either normal or as an attack, with exactly one specific attack type. }

 {There are in total $N= 25193$ pieces of connection records in our selected dataset, of which $11743$ are intrusions. We assume abnormal connections are initialed by abnormal users, and vice versa. Thus we can estimate a prior belief regarding the type of users as 
\begin{equation}
    p(H_1) = \frac{11743}{25193}\approx 0.466,\; p(H_0)  \approx 0.534.
\end{equation}}

 {In general, each connection record of TCP dump data contains many features such as source bytes, destination bytes, protocols, log-ins, server error rates, number of shells, etc. For convenience, we assume here that the detector determines the type of the user based on distributional information of the `log-in' feature, which is characterized as the random variable $\hat{m}$ that takes binary values: $\hat{m}=0$ if the log-in attempt is successful and $\hat{m}=1$ otherwise. Normal users and abnormal users have log-in attempts with failure rates of $\tha_0$ and $\tha_1$, respectively. Accordingly, the detector forms two hypotheses as follows:
\begin{equation}
H_0: \hat{m}\sim f_0(m),\;\;H_1: \hat{m}\sim f_1(m),
\end{equation}
with
\begin{equation}
\begin{aligned}
    f_1(1) = P[m=1|H_1] &= \tha_1,\; f_1(0) = P[m=0|H_1] = 1-\tha_1,\; \\ f_0(1) = P[m=1|H_0] &= \tha_0,\; f_0(0)= P[m=1|H_0] = 1-\tha_0.
    \label{hypo:binom}
\end{aligned}
\end{equation} }

 {Now we consider the proposed framework depicted in Figure \ref{fig:naive_detector}. An intrusion detection system (IDS, she) aims to identify whether a user (It) is normal $H_0$ or abnormal 
$H_1$ by analyzing the `log-in' extracted from the data packets obeying distributions specified in \eqref{hypo:binom}. We observe from the dataset that among the packets from abnormal users only $409$ out of $11743$ log-ins were successful, while from normal users $9536$ out of $13450$ were successful.  We thus estimate from the dataset that the failure rates under different hypotheses are $\tha_1 = 0.292,\tha_0 = 0.966$. The attacker intercepts data packets in the middle, falsifies the log-in features, and transmits the modified data packets to the IDS according to his equilibrium strategies as in \eqref{eq:signal_equil_0} and \eqref{eq:signal_equil_1}. The evasion-aware IDS makes decisions based on the attacker's message according to \eqref{eq:receiver_np}. }

 As another comparison, a robust IDS  \cite{fauss2016robust_band_hypo_testing} derives an optimal strategy $\hat{\delta}^*\in 2^M$ by solving the minimax problem: $$\underset{\hat{\delta}\in 2^M}{\min}\underset{(\hat{\sigma_1},\hat{\sigma}_0)\in \hat{A}_S}{\max}\;\beta \mE_{\hat{\sigma}_0}[\hat{\delta}(m)] +\; \mE_{\hat{\sigma}_1}[1- \hat{\delta}(m)],$$
with the uncertainty set $\hat{A}_S = \{(\hat{\sigma}_1,\hat{\sigma}_0):\;\underline{c}f_0\leq \hat{\sigma}_0\leq \bar{c}f_0,\;\underline{c}f_1\leq \hat{\sigma}_1\leq \bar{c}f_1,\;\underline{c}<1<\bar{c}\}$ and the parameter $\beta$ the same as in \eqref{eq:signal_equil_0} and \eqref{eq:signal_equil_1}. The optimal strategy $\delta^*$ applies a likelihood ratio test for `least favorable distributions (LFD)', denoted as $f^*_0,f^*_1$: that is, $\hat{\delta}^*(m) = 1$ if $f^*_1(m)>\beta f^*_0(m)$
and $\hat{\delta}^*(m) = 0$ otherwise. The robust IDS is also evasion-aware, taking into consideration the potential manipulation of messages by the attacker, yet not proactive, as it assumes that the attacker directly antagonizes her despite the user's type or the attacker's message. {To demonstrate the generalizability of the proposed approach, we also conduct another case study on the IDS dataset \cite{ids2018}, and the results are provided in the supplementary document associated with this paper.}

\begin{figure}
    \centering
\includegraphics[scale=0.35]{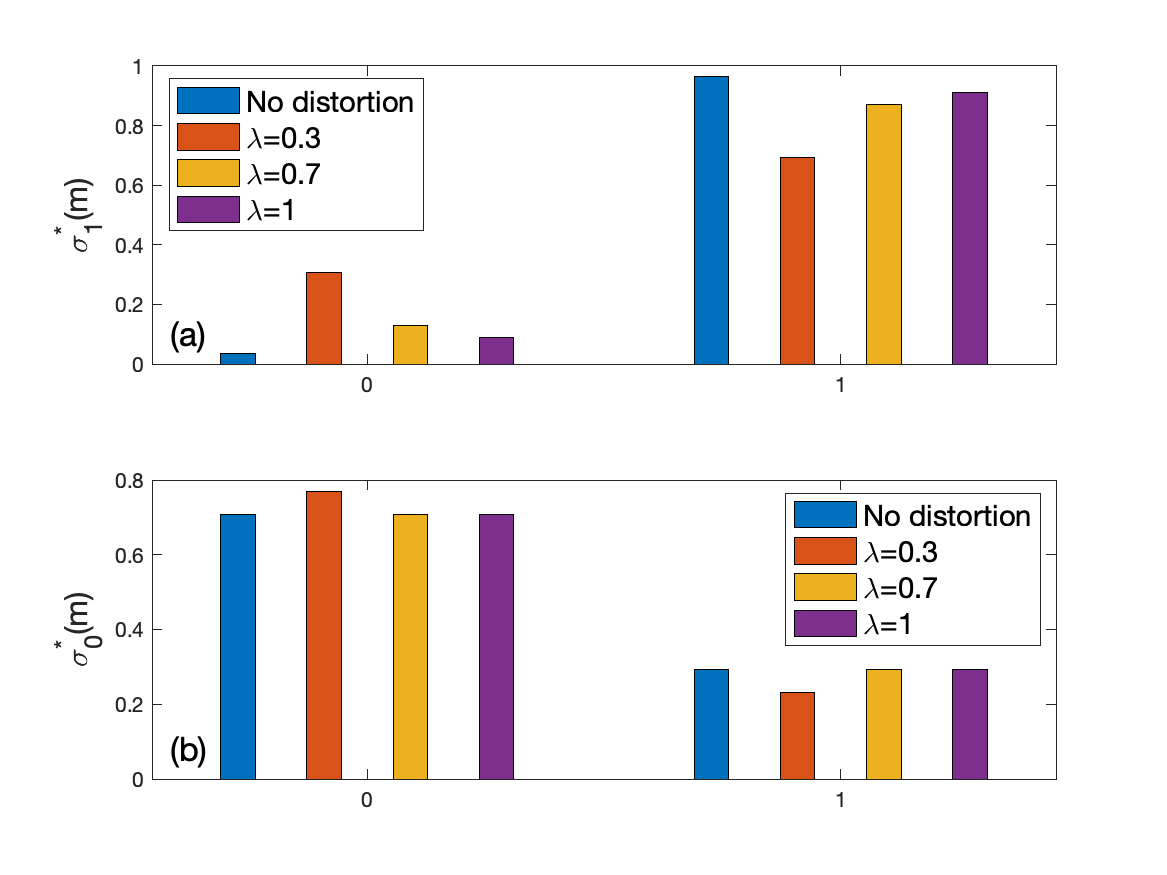}\\
    \includegraphics[scale=0.175]{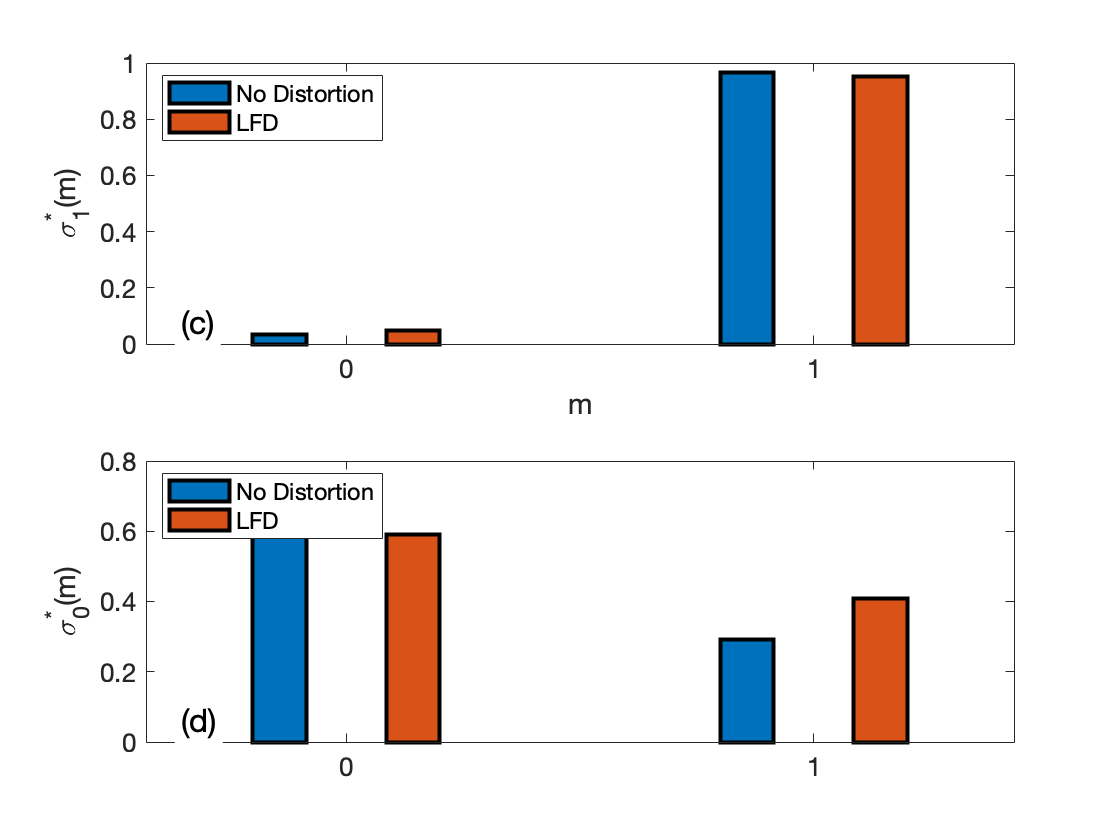}
    \caption{In (a)(b) we depict the attacker's equilibrium strategies \eqref{eq:signal_equil_0}, \eqref{eq:signal_equil_1} against an evasion-aware IDS; in (c)(d) we illustrate the attacker's equilibrium strategies against a robust IDS \cite{fauss2016robust_band_hypo_testing}.   }
    \label{fig:strategy_binomial}
\end{figure}
\begin{figure}
    \centering
    \includegraphics[scale=0.35]{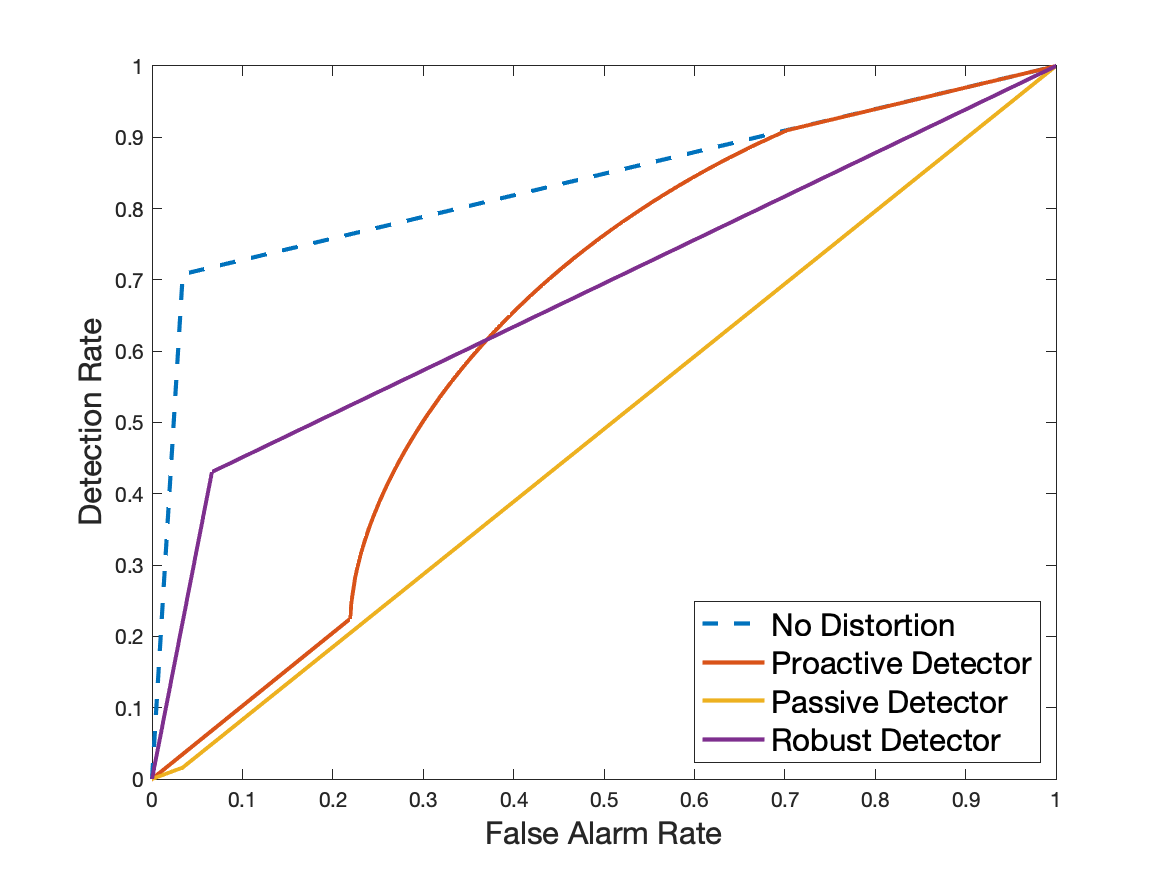} 
    \caption{EROC curves of a non-adversarial IDS, a non-strategic IDS, an evasion-aware IDS, and a robust IDS at equilibrium strategies. The true distributions under each hypothesis are both Bernoulli distributions specified in \eqref{hypo:binom} with parameters  $\tha_1 = 0.966$ and $\tha_0 = 0.292$. The parameter $\lambda=0.2$.}
    \label{fig:roc_bern}
\end{figure}
\begin{figure}
    \centering
     \includegraphics[scale=0.35]{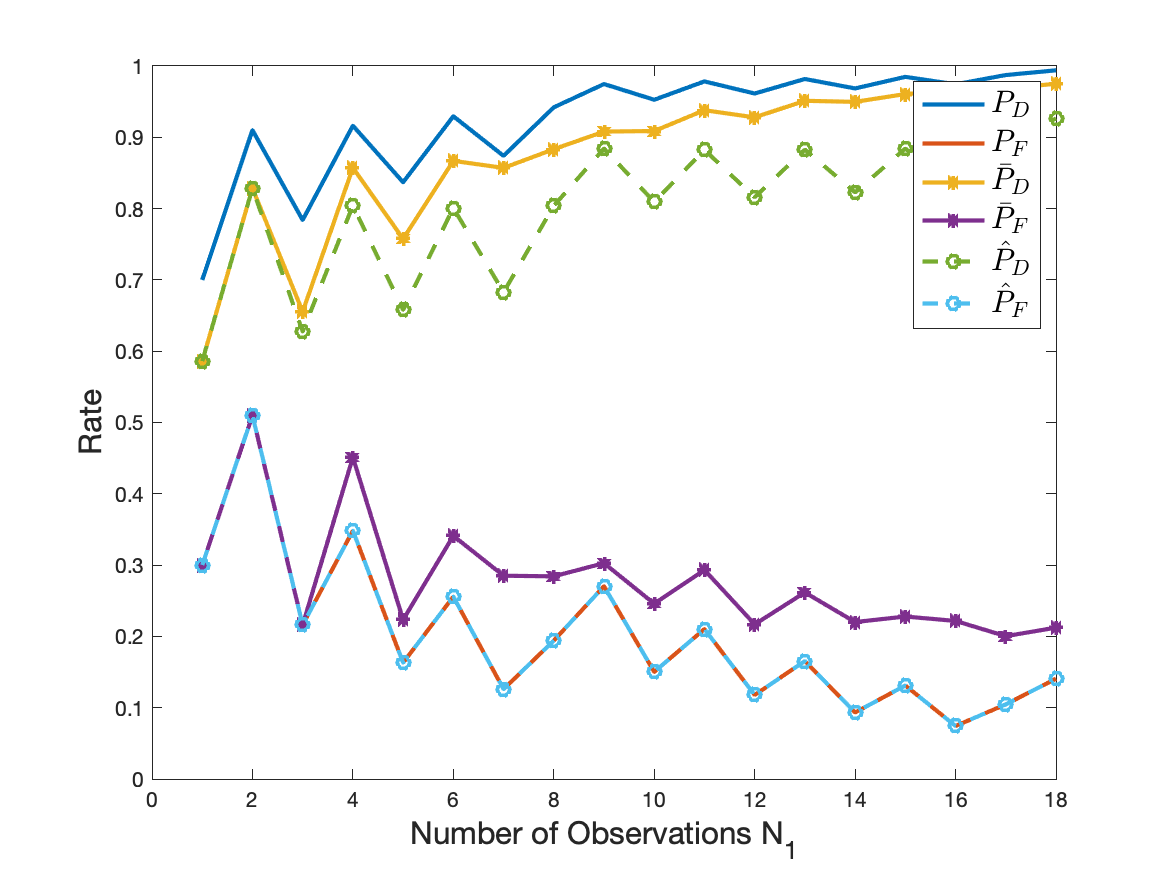}
    \caption{The detection rates and false alarm rates for passive, proactive, and non-adversarial IDS under $N_1=15$ observations with the initial threshold $\beta=0.9$. The nominal distributions under $H_1, H_0$ are Bernoulli distributions with parameter $\tha_1 =0.7,\tha_0=0.3$. The parameter $\lambda = 0.75$.}
    \label{fig:repeated_bern}
\end{figure}

\subsection{Numerical Results}

We set $\beta = 6.533$ and $\tha_1,\tha_0, f_1,f_0$ as introduced at the beginning of this section and depict in Fig. \ref{fig:strategy_binomial} the strategies of man-in-the-middle at the equilibria for evasion-aware detectors under different choices of $\lambda$. {We observe that as $\lambda$ decreases (i.e., the penalty for distorting messages is reduced), the attacker's equilibrium strategies increasingly favor a higher probability of selecting $m=0$ when the user is abnormal. Similarly, when the user is normal, the attacker's equilibrium strategies also tend to assign a higher probability to $m=1$, although this effect is less pronounced.} {In Fig. \ref{fig:strategy_binomial}(c)(d) we depict the attacker's equilibrium strategies against a robust detector.  We set $\underline{c}=0.5,\;\bar{c} =1.3$ and $\beta =6.533$ for a robust detector.  We observe that, when confronting a robust detector, the attacker tends to adjust the likelihood ratio to be closer to $1$.} The resulting equilibrium strategies of the attacker can help compute the detection rate $\bar{P}_D$ and the false alarm rate $\bar{P}_F$ indicated in \eqref{false_alarm_rate_signalling}.  In Fig. \ref{fig:roc_bern}, we depict the performances of an evasion-aware detector and other counterparts with different rational capacities using EROC curves. We observe that the evasive-aware IDS outperforms the passive IDS in the sense that the EROC curve of the evasive-aware IDS lies above other counterparts. {Additionally, also, the proactive IDS may outperform a robust IDS due to the differing behaviors of attackers when facing each type of detector. Against a robust IDS, attackers tend to make thbecause attackers behave differently when facing each different detectors: they make distributions nearly almost identical, while they against a robust IDS, but expand the uncertainty regions when facing against a proactive IDSone. As the false alarm rate increases, the detection rate of a robust IDS grows more slowly because samples with lower likelihood ratios—those, less likely to belong to the target distribution—are, get rejected. In contrast, the proactive IDS, which dynamically adjusadapts its threshold, may dynamically, experiences slower detection rate growth when distributions are heavily suppressed, but it can surpass the performance of the robust IDS in certain regimes.} {In Fig. \ref{fig:strategy_binomial}(c)(d) we depict the attacker's equilibrium strategies against a robust detector.  We set $\underline{c}=0.5,\;\bar{c} =1.3$ and $\beta =6.533$ for a robust detector.  We observe that, when confronting a robust detector, the attacker tends to adjust the likelihood ratio to be closer to 1.} The resulting equilibrium strategies of the attacker can help compute the detection rate $\bar{P}_D$ and the false alarm rate $\bar{P}_F$ indicated in \eqref{false_alarm_rate_signalling}.  In Fig. \ref{fig:roc_bern}, we depict the performances of an evasion-aware detector and other counterparts with different rational capacities using EROC curves. We observe that the evasive-aware IDS outperforms the passive IDS in the sense that the EROC curve of the evasive-aware IDS lies above other counterparts. {Additionally, a proactive IDS may outperform a robust IDS due to the differing behaviors of attackers when facing each type of detector. Against a robust IDS, attackers tend to make the distributions nearly identical, while they expand the uncertainty regions when facing a proactive IDS. As the false alarm rate increases, the detection rate of a robust IDS grows more slowly because samples with lower likelihood ratios—those less likely to belong to the target distribution—are rejected. In contrast, the proactive IDS, which dynamically adjusts its threshold, may experience slower detection rate growth when distributions are heavily suppressed, but it can surpass the performance of the robust IDS in certain regimes.} Asymptotically, as $\lambda$ increases, the attacker tends to distort the message less, and the detector's EROC curve grows more closely to the one without distortion (the blue curve in Fig. \ref{fig:roc_bern}) and vice versa. 

The detection rates and false alarm rates for the detectors under repeated observations are depicted in Fig. \ref{fig:repeated_bern}. We could observe that just like the non-adversarial IDS, both the passive and the proactive IDS can also reach a perfect prediction (detection rate converges to 1, while false alarm rate converges to 0) given more and more distorted messages, yet the convergence speeds of these detectors vary: compared to the non-adversarial IDS, the passive IDS converges more slowly on its detection rates, while the proactive IDS converges more slowly on its false alarm rates.

    \subsection{Multi-Stage Analysis}
  We further pursue and seek multi-stage analysis. The multi-stage analysis aim to show the advanced persistent threats (APT) \cite{chen2014study_advanced_persistent_threat} that have been receiving more and more attention over the years. 
    The temporal correlation of advanced persistent threats involves examining the timeline of the attack, identifying its stages and understanding how threat actors adapt and evolve their tactics over time. The temporal perspective is crucial for cybersecurity professionals to gain insights into the motivations, objectives of intrusion detection evasion attacks. 

    In Figure \ref{fig:multi_stage}, using the algorithm proposed in \cite{huang2019dynamic_bayesian_game}, we plot the attacker's equilibrium strategies over the stages of a message sequence. We observe that the attacker's state-dependent equilibrium strategies are asymptotically closer to the true distributions produced by users under each hypothesis no matter the choice of $\lambda$, suggesting that the proactive detector can find out the true type of the user and the attacker's mitigation will have less and less influence upon the detector's performance over the stages. Such results imply that the proposed proactive detector has the potential detecting and combatting persistent intrusion evasion attacks.

\begin{figure}
    \centering
\includegraphics[scale=0.17]{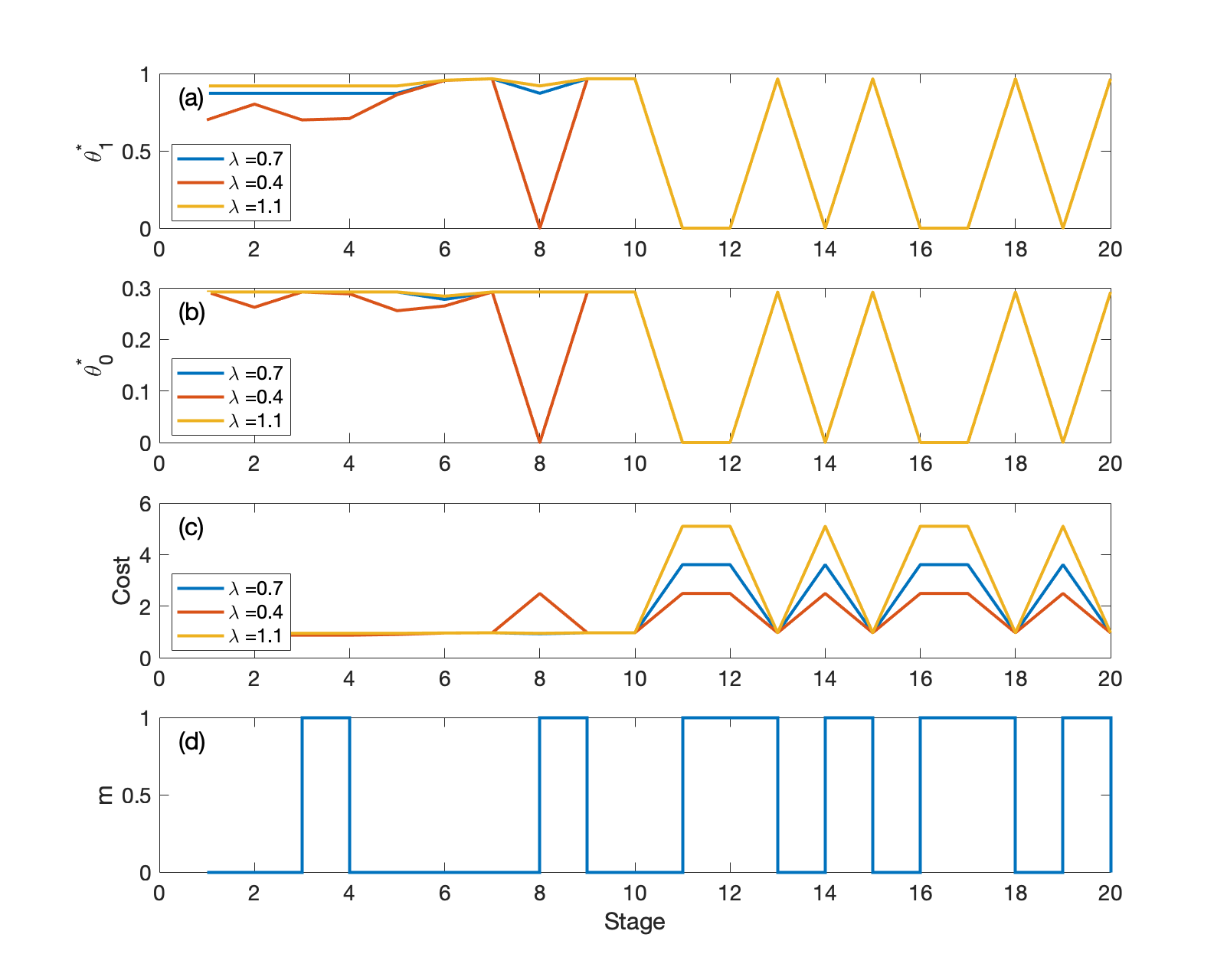}
    \caption{ {(a)(b)The evolution of attacker's equilibrium strategies $\sigma^*_1,\sigma^*_0$ (represented by $\tha^*_1,\tha^*_0$ with meaning in \eqref{hypo:binom}) in multi-stage game.  {(c) the attacker's cost function over the stages. (d) the message trajectory over the stages.} We set the time horizon $N_1=20$ and keep the initial setting of the signaling game as before.}}
    \label{fig:multi_stage}
\end{figure}

\section{Conclusion} 
\label{sec:conclusion}
We have consolidated the classical Neyman-Pearson detection method into the framework of sequential games to develop holistic theories of passive detectors and of proactive detectors to help develop evasion-aware detectors that cope with ever-developing stealthy attackers. Depending on different rational capacities for the detectors, we employ the frameworks of Stackelberg games and of signaling games to capture the behaviors of passive and of proactive detectors and depict the detectors' performances evaluated by EROC curves. Our results reveal the proactive detectors outperform the passive ones as the former ones can take advantage of the distorted message to infer the user's type so as to countermeasure the attacker's manipulation of messages.  

Future extensions may include designing a decentralized detection scheme for M-ary hypothesis testing similar to \cite{tsitsiklis1993decentralized}. The signaling game framework can also be combined with statistical estimation techniques, which also has implementations in security such as analyzing spoofing attacks \cite{d2021_spoof_attack}. 

\bibliographystyle{ieeetr}
\bibliography{rl}

 \appendix
 
\section{Appendix}
\subsection{Proof of Proposition \ref{prop:sol_signaling_lam_0}
}
\label{subsec:lam_0_sender}
\begin{proof}
When $\lambda = 0$ our assumption \eqref{assume:full_support} does not need to hold since there is no KL divergence term in the objective function. Thus we can assume the space of attacker's strategies $\sigma_1,\sigma_0$ are all probability distributions on the message space $M$ so the optimization problem \eqref{eq:potential} reduces to the following:
\begin{equation}
    \underset{\substack{(\sigma_1,\sigma_0)\in \bar{A}_S}}{\min}\;\int_{m\in M_1}{\sigma_1(m)dm},
\end{equation}
where $\sigma_1,\sigma_0$ are subjective to the constraints the same as \eqref{eq:potential}. When $\beta\geq 1$, we verify that all strategies $\sigma_1\leq \beta \sigma_0$ lead to the region of rejection $M_1$ to be an empty set, which means the value of the objective function is $0$. Because the objective function value cannot be negative (it means receiver's detection rate), all of such pairs of strategies are equilibrium strategies. 

When $0<\beta<1$, we can find the attacker's equilibrium strategies by consider best response dynamics. Denote $\text{BR}_0,\text{BR}_1: L^1(M)\rightarrow L^1(M)$ as the best responding strategies $\sigma_0,\sigma_1$ when the other is given. We can construct the following best responses:
\begin{align}
\text{BR}_0[\sigma_1] & = \begin{cases}
\frac{1}{\beta}\sigma^*_1(m) &\text{if}\; \sigma_1(m)>\tau,\;\int_{\sigma_1(m)>\tau}{\sigma_1(m) dm} = \beta, 
\label{br0}\\
0 & \mbox{otherwise}.
\end{cases} \\
\text{BR}_1[\sigma_0] \;&= \begin{cases}
\beta \sigma_0(m)& \text{if}\; \sigma_0(m)\geq 0, \\
(1-\beta)\varphi(m) & \mbox{otherwise}.
\end{cases}
\label{br1}
\end{align}
Now we can verify that the equilibrium \eqref{eq:sig_s_h1} and \eqref{eq:sig_s_h0} meet the requirement $\sigma_1^* = \text{BR}_0[\sigma^*_0],\;\sigma_0^* = \text{BR}_1[\sigma^*_1]$, which concludes the proof.
\end{proof}

\subsection{Proof of Proposition \ref{prop:sol_signaling_lam_infty}
}
\label{subsec:lam_infty_sender}
\begin{proof}
We prove the statement by construction. When $\lambda\rightarrow \infty$, the objective function in \eqref{eq:potential} reduces to the sum of KL divergence between $\sigma_1$ and $f_1$ and the KL divergence between $\sigma_0$ and $f_0$. By the non-negativity of KL divergences, we know both terms are non-negative and only reach zero when \eqref{eq: kl_divgence_min} holds,  which can be realized by setting $\sigma^*_1(m)=f_1(m),\;\sigma^*_0(m)=f_0(m)$. As a result, the optimal solutions for the problem \eqref{eq:potential} when $\lambda\rightarrow \infty$ is \eqref{eq: kl_divgence_min}. 
\end{proof}

\subsection{Proof of Theorem \ref{prop:finite_lam_signaling}}
\label{sec:sender_optimal_strategy_finite_lambda}

\begin{proof}
To solve \eqref{eq:potential}, we first introduce  $\sigma'_1,\sigma''_1,\sigma'_0,\sigma''_0: M\rightarrow \mR_+,$ as the restrictions of $\sigma_1,\sigma_0$ on $\bar{M}_0,\bar{M}_1$, respectively, where $\bar{M}_0 = \left\{m : \sigma_1(m)\leq \beta\sigma_0(m)\right\},\;\bar{M}_1 = \left\{m : \sigma_1(m)> \beta\sigma_0(m)\right\}$. 
Based on the decomposition we can convert \eqref{eq:potential} into the following:
\begin{equation}
\begin{aligned}
&\underset{\substack{\sigma'_0,\sigma'_1,\sigma''_0,\sigma''_1}}{\min}\int_{M}{\sigma''_1(m)dm}\\
 &+ \lambda \int_{M}{\sigma'_1(m)\ln\frac{\sigma'_1(m)}{f_1(m)} + \sigma''_1(m)\ln\frac{\sigma''_1(m)}{f_1(m)}dm}
  \\
  &+\lambda \int_{M}{\left(\sigma'_0(m)\ln\frac{\sigma'_0(m)}{f_0(m)} + \sigma''_0(m)\ln\frac{\sigma''_0(m)}{f_0(m)}\right)dm},
  \\
    &\text{s.t.}\;\int_{M}{(\sigma'_0(m) + \sigma''_0(m))dm} = 1, \\
    & \qquad \int_{M}{(\sigma'_1(m) + \sigma''_1(m))dm} = 1, \\
    & \qquad {\sigma''_1(m)}-\beta{\sigma''_0(m)}> 0,\;\forall m\in \bar{M}_1,\\
    &\qquad {\sigma'_1(m)}- \beta{\sigma'_0(m)}\leq 0,\;\forall m \in \bar{M}_0, \\
    &\qquad \sigma'_1(m),\sigma_1''(m),\sigma'_0(m),\sigma''_0(m)\geq 0,\;\forall m\in M.
    \label{opt:function1}
\end{aligned}
\end{equation}
where the support for $\sigma'_0,\sigma'_1,\sigma_0'',\sigma''_1\in L^1(M)$ satisfy $\supp \sigma'_1 = \supp \sigma'_0,\;\;\supp \sigma''_1=\supp \sigma''_0,\;\supp \sigma'_1  \cap \supp \sigma''_1 = \varnothing,\;\supp \sigma'_1  \cup \supp \sigma''_1 = M$.
By theories of equivalent optimization problems \cite{boyd2004convex_optimization}, we could see the optimal solutions $\sigma^*_0,\sigma^*_1$ of the problem \eqref{eq:potential} and the one of \eqref{opt:function1} can be connected using decomposition.

Solving the KKT conditions of the problem \eqref{opt:function1} we write $\sigma'_1,\sigma'_0,\sigma''_1,\sigma''_0$ in terms of $\rho_0,\rho_1$: $\forall m\in \bar{M}_0$,
\begin{align}
&{\sigma'^*_0}(m) = c_0f_0(m) e^{\frac{\beta\rho_0(m)}{\lambda}},\;{\sigma'^*_1}(m) =c_1f_1(m)e^{-\frac{\rho_0(m)}{\lambda}}; \\
&\forall m\in \bar{M}_1,\;\sigma''^*_0(m) = c_0f_0(m),\; \sigma''^*_1(m) =c_1 f_1(m)e^{-\frac{1}{\lambda}},
    \label{sigma1_m1}
\end{align}
where we introduce $c_0 = e^{-\gamma_0/\lambda},\;c_1 = e^{-\gamma_1/\lambda}$. We can then construct the optimal solution for \eqref{eq:potential} via $\sigma^*_1 = \sigma'^*_1 + \sigma''^*_1,\;\sigma^*_0 = \sigma'^*_0 + \sigma''^*_0$.

We know $\rho_0(m)\geq 0$ for every $m\in\bar{M}_0$. Denote $M_* = \{m:\;\rho_0(m)>0\}$ and $M_0 = \{m:\;\rho_0(m)=0\}$ as a partition of $M_0$.  For every $m\in {M}_*$, we have
\begin{equation}
    \sigma^*_1(m)-\beta\sigma^*_0(m)=0
    \label{eq:sigma_1_beta_sigma0}
\end{equation} by complementary slackness. Similarly we know for every $m\in M_0$, $\sigma^*_1(m)-\beta\sigma^*_0(m)>0$. Then, we have  
\begin{align}
    \sigma^*_0(m) & = \begin{cases}
    c_0f_0(m), &m\in M_0, \\
     c_0f_0(m)e^{\frac{\beta\rho_0(m)}{\lambda}}, &m\in M_*, \\
     c_0f_0(m), &m\in M_1.
    \end{cases}
    \label{eq:sigma0_final}\\
    \sigma^*_1(m) & = \begin{cases}
    c_1f_1(m), &m\in M_0, \\
     c_1f_1(m)e^{-\frac{\rho_0(m)}{\lambda}}, &m\in M_*, \\
     c_1f_1(m)e^{-\frac{1}{\lambda}}&m\in M_1,
    \end{cases}
    \label{eq:sigma1_final}
\end{align}
which leads to
\begin{align}
    {c_1}{f_1(m)}e^{-\frac{\rho_0(m)(1+\beta)}{\lambda}} &= {c_0}{\beta f_0(m)}.
    \label{eq:rho_equation}
\end{align}
Since the LHS of \eqref{eq:rho_equation} depends on $m$ explicitly and the RHS of \eqref{eq:rho_equation} does not, the derivative of the LHS of \eqref{eq:rho_equation} with respect to $m$ must vanish and we can express $\rho_0$ in terms of a constant $\zeta$ as follows:
  \begin{align}
     \rho_0(m) &= \frac{\lambda}{1+\beta}\left(\ln \left(\frac{f_1(m)}{f_0(m)}\right) + \ln \zeta \right),\;m\in M_*. \label{sol:rho0}
 \end{align}
 Substituting \eqref{sol:rho0} into \eqref{eq:rho_equation} gives us
  \begin{align}
      {c_1}  = {\beta}\zeta{c_0}.
      \label{a0a1c0c1}
  \end{align}
Combining \eqref{a0a1c0c1}, \eqref{eq:sigma0_final} and \eqref{eq:sigma1_final} and \eqref{eq:sigma_1_beta_sigma0}, we can find the implicit dependence of $\zeta$ on $\beta$ as indicated in \eqref{zeta_beta_1}, \eqref{zeta_beta_2}. By assumption \ref{lemma:function_J_zeta} we know such $\zeta$ exists. 
The final result follows from substituting \eqref{eq:rho_equation} in \eqref{eq:sigma0_final} and \eqref{eq:sigma1_final}.
\end{proof}

\subsection{Discussions on Multi-stage Game}
To prove proposition \ref{prop:s_pbne}, we first introduce some assumptions. The following assumption states that the only useful information accumulated throughout stages is the attacker's mixed strategies of messages. 
    \begin{assume}[Action-independent assumption]
At every stage $j\in[N_1]$, the proactive detector's optimal decision rule $\delta^{(j)*}(\cdot|h^{(j)})\in \bar{\Gamma}^{(j)}$, the attacker's optimal mixed strategies of generating manipulated messages $\sigma^{(j)*}_1(\cdot|h^{(j)}),\sigma^{(j)*}_0(\cdot|h^{(j)})\in \bar{A}^{(j)}_S $ and the posterior belief $p(\cdot|h^{(j)})$ depend only on the attacker's history of mixed strategies. Specifically, we have for $k\in\{0,1\}$,
\begin{align}
    {\bar{\delta}}^{(j)*}(m_j|h^{(j)}) &= {\bar{\delta}}^{(j)*}(m_j|m^{(j)}),
    \label{action_hisotry_dependent}\\
     \sigma^{(j)*}_k(m_j|h^{(j)}) &=     \sigma^{(j)*}_k(m_j|m^{(j)}),
     \label{msg_history_dependent}\\
    p(H_k|h^{(j)}) &=  p(H_k|m^{(j)}).\;
    \label{belief_system}
\end{align}
\end{assume}

 At the stage $j$, the proactive detector faces the following optimization problem:
\begin{equation}
\begin{aligned}
\underset{\bar{\delta}^{(j)}\in\bar{\Gamma}^{( j)}}{\sup}&\;\bar{\delta}^{(j)}(m^{(j)}){p(H_1|m^{(j)})}, \;\text{s.t.}&\bar{\delta}^{(j)}(m^{(j)}) p(H_0|m^{(j)}) \leq \alpha_j.
     \label{lagrangian_np}
\end{aligned}
\end{equation}    
Denote a generic decision rule at stage $j$ as $\bar{\delta}^{(j)}:M^{\otimes j}\rightarrow [0,1]$. Then the space of decision rules at the stage $j$, denoted as $\bar{\Gamma}^{(j)}$, can be written as follows:
\begin{equation}
    \bar{\Gamma}^{(j)} = \{\delta^{(j)}\in 2^{M^{\otimes j}}:\;\delta^{(j)}(m^{(j)})p(H_0|m^{(j)})\leq \alpha_j\}.
    \label{def:detector_multistage_strategy_j}
\end{equation}
    At every stage $j$ the attacker faces the following optimization problem:
\begin{equation}
\begin{aligned}
&\underset{\substack{(\sigma^{(j)}_1,\sigma^{(j)}_0)\in \bar{A}^{(j)}_S}}{\min}\; \int_{\delta^{(j)*}(m_j)=1}{\sigma^{(j)}_1(m_j)dm_j} \\
   &+ \lambda  \int_{ M}{\left({\sigma^{(j)}_1(m_j)\ln\frac{\sigma^{(j)}_1(m_j)}{f_1(m_j)}} + \sigma^{(j)}_0(m_j)\ln\frac{\sigma^{(j)}_0(m_j)}{f_0(m_j)}\right)dm_j}, 
    \label{eq:potential_N}
\end{aligned}
\end{equation}

     At the $j$-the stage, we can denote $\sigma^{*(j)}_1, \sigma^{*(j)}_0 \in \bar{A}^{(j)}_S$ as his generic mixed strategies under user's type $H_1, H_0$ respectively. 
 The attacker can aim at minimizing the sum (or average) of the cost function as in \eqref{eq:potential_N}. For proposition \ref{prop:s_pbne} we compute the subgame perfect Bayesian Nash equilibrium (s-PBNE): we assume that at any stage $j$, both the attacker and the proactive detector aim at optimizing their strategies as if the game starts at stage $j$. Proposition \ref{prop:s_pbne} follows by applying forward induction and the proof of theorem \ref{prop:finite_lam_signaling}.

\end{document}